\RequirePackage{amsthm}
\documentclass[sn-mathphys-num]{sn-jnl}

\usepackage{amsfonts}
\usepackage{amsmath}
\usepackage{amsthm}
\usepackage{xcolor}
\usepackage{amssymb}
\usepackage{graphicx}%
\usepackage{float}
\usepackage{soul}

\usepackage[linesnumbered,boxed]{algorithm2e}
\SetKwRepeat{Do}{do}{until}

\renewcommand{\P}{\mathcal{P}}

\renewcommand{\setminus}{-}

\theoremstyle{thmstyleone}

\newtheorem{theorem}{Theorem}
\newtheorem{corollary}{Corollary}
\newtheorem{lemma}{Lemma}

\title{Path Partitions of Phylogenetic Networks}

\author*[1]{\fnm{Manuel} \sur{Lafond}}\email{manuel.lafond@USherbrooke.ca}

\author*[2]{\fnm{Vincent} \sur{Moulton}}\email{v.moulton@uea.ac.uk}

\affil*[1]{\orgdiv{Department of Computer Science}, \orgname{Université de Sherbrooke}, \orgaddress{\country{Canada}}}

\affil*[2]{\orgdiv{School of Computing Sciences}, \orgname{University of East Anglia}, \orgaddress{Norwich, NR4 7TJ, \country{United Kingdom}}}

\abstract{
In phylogenetics, evolution is traditionally represented in a tree-like manner.  However, phylogenetic networks can be 
more appropriate for representing evolutionary events such as hybridization, horizontal gene transfer, and others.  
In particular, the class of forest-based networks was recently introduced to 
represent introgression, in which genes are swapped between species.  
A network is forest-based if it can be obtained by adding arcs to 
a collection of trees, so that the endpoints of the
new arcs are in different trees.
This contrasts with so-called tree-based networks, which are formed by adding arcs 
within a single tree.  

We are interested in the computational complexity of recognizing forest-based networks, which 
was recently left as an open problem by Huber et al.  It has been observed that forest-based networks coincide 
with directed acyclic graphs that can be partitioned into induced paths, each 
ending at a leaf of the original graph.  Several types of path partitions 
have been studied in the graph theory literature, but to our 
best knowledge this type of `leaf induced path partition' has 
not been directly considered before. The study of forest-based networks in terms of these 
partitions allows us to establish closer relationships between phylogenetics and algorithmic 
graph theory, and to provide answers to problems in both fields.

More specifically, we show that deciding whether a network is forest-based is NP-complete, even on input networks that are tree-based, binary, and have only three leaves.  
This shows that partitioning a directed acyclic graph into a constant number of 
induced paths is NP-complete, answering a recent question of Fernau et al.  
We then show that the problem is polynomial-time solvable on binary networks with 
two leaves and on the recently introduced class of orchards, which 
we show to be always forest-based.  Finally, for undirected graphs, 
we introduce unrooted forest-based networks and provide 
hardness results for this class as well.
}


\begin{document}

\maketitle

\section{Introduction}

Recently, there has been growing interest
in using networks in addition to rooted trees 
to represent evolutionary histories of species \cite{kong2022classes}.
Formally, a \emph{network} is a connected, directed acyclic graph (DAG) $N$ 
in which the set $L(N)$ of sinks or \emph{leaf set}  
corresponds to a collection of species.
Much work to date has focused on networks having single root,
although recent work has also considered networks that have multiple roots \cite{scholz2019osf}.
Networks are commonly used to model the evolution of species which 
undergo various forms of \emph{reticulate evolution} \cite{sneath1975cladistic} (e.g. where
species come together or hybridize to form a new species), and 
several classes of networks have been defined in recent years that have been intensively studied in
the literature (see e.g. \cite{kong2022classes} for recent survey).

Networks that have a single source or root 
are usually called \emph{phylogenetic} networks, 
and much of the work on these has 
focused more specifically on \emph{binary} phylogenetic networks, in which all of the leaves 
have indegree one, the root has outdegree 2, and all other 
vertices have total degree 3 (see e.g. \cite[Chapter 10]{steel2016phylogeny}).
Note that phylogenetic networks whose underlying undirected graph is a
tree are called \emph{phylogenetic trees}, which are 
better known as the evolutionary trees that often appear in biology textbooks.
One special class of phylogenetic networks that has recently received
considerable attention are the \emph{tree-based} networks \cite{francis2015phylogenetic,jetten2016nonbinary}.
These are essentially phylogenetic networks that can be
formed by adding a collection of arcs to a phylogenetic tree. 

Binary tree-based phylogenetic networks have several characterizations, one of which is as follows.
Given a directed acyclic graph $N$, we define a \emph{leaf path partition} 
of $N$ to be a collection of directed paths in $N$ that partition
the vertex set of $N$ and such that each path ends in a leaf (or sink) of $N$.
In \cite[Theorem 2.1]{francis2018new}, it is shown that a binary phylogenetic
network is tree-based if and only if it admits a leaf path partition.
More recently, it has been noted that leaf path partitions
also naturally arise when considering the 
closely related class of \emph{forest-based networks} \cite{huber2022forest}.
These are networks which can be formed by adding
arcs to a {\em collection} of phylogenetic trees, or \emph{phylogenetic forest},
so that each added arc has its end vertices in different trees, and 
the network so obtained is connected. In \cite{huber2022forest}
it is shown that a network is forest-based if and only
if it admits an \emph{induced} leaf path partition, that 
is, a path partition in which each directed path is an induced path.

By exploiting
leaf path partitions, we shall focus on answering some 
complexity problems concerning forest-based 
and other closely related networks. 
Note that, due in part to 
their various applications in mathematics
and computer science, path partitions of graphs have been extensively studied 
(see e.g. \cite{manuel2018revisiting}), and 
they remain a topic of current interest. 
For example, in 
the recent paper \cite{fernau2023parameterizing} the complexity of several
path partition problems of digraphs, as well in graph in general, are surveyed and determined. 
However, although path partitions of 
directed acyclic graphs in which each path contains 
one vertex from a fixed subset of the vertex
set have been considered (see e.g. \cite{sambinelli2022alpha} and the
references therein where they are called $S$-path partitions for a
fixed subset $S$ of the vertex set),
to our best knowledge the concept of leaf path partitions
for general directed acyclic graphs appears to be new.

The main problem that we will consider in this
paper is determining the complexity of deciding whether a DAG 
admits a leaf induced path partition (leaf IPP for short), as
well as the closely related problem of deciding whether a network is forest-based or not.
In \cite{huber2023network}, using graph colourings, it is shown that 
it is NP-complete to decide if a binary, tree-child\footnote{A network
is tree-child if each
non-leaf vertex has at least one child with indegree 1.} 
network $N$ with a fixed number of roots $k\ge 3$ 
is \emph{proper forest-based}, 
that is, if $N$ can be constructed from a phylogenetic 
forest with $k$ components as described above. 
However, the problem of deciding if a network $N$ is forest-based is left 
as an open problem. Also very closely related is the recent work 
of Fernau et al.~\cite{fernau2023parameterizing},
where amongst other results they show that deciding whether a binary planar DAG can be partitioned into at most $k$ 
induced paths, for given $k$, is NP-complete, and also that this problem is $W[1]$-hard 
on DAGs for parameter $k$.  Among their open questions, they ask whether the 
problem is in XP for parameter $k$.  In other words, they ask whether 
it is possible to achieve time complexity $n^{f(k)}$, which would be polynomial if $k$ is a constant.

In this paper, we shall show that deciding whether a network $N$ is forest-based is NP-complete even 
in case $N$ is a binary, tree-based phylogenetic network.
A key component in our proof is to show
that it is NP-complete to decide if a directed acyclic graph
with three roots and three leaves admits a leaf IPP, 
which we do by 
reducing from \textsc{Monotone NAE-3-SAT} \cite{dehghan2015complexity}.
This implies that it is NP-complete to decide whether a binary DAG can be split into $k = 3$ induced paths, 
which thus also answers the question from~\cite{fernau2023parameterizing}
mentioned in the last paragraph on XP membership.  Our reduction produces networks 
of linear size, which also implies that under the \emph{Exponential Time Hypothesis} (ETH), no 
sub-exponential time algorithm is possible for the forest-based recognition problem on three leaves.  
Recall that the ETH states that, in particular, 3-SAT cannot be solved in time $2^{o(n + m)} n^c$, 
with $n, m$ the number of variables and clauses, respectively, and $c$ is any 
constant~\cite{impagliazzo2001problems}, and that the lower bound 
applies to \textsc{Monotone NAE-3-SAT}~\cite{antony2024switching}.

We also show that one can decide in polynomial time whether a binary DAG can be split 
into $k = 2$ induced paths, by a reduction to 2-SAT.  The case of $k = 2$ and non-binary DAGs remains open. 
As an additional positive result, we show that all the networks that belong to the 
well-known class of \emph{orchards} are forest-based~\cite{erdHos2019class,van2022orchard}.  
Orchards are networks that are consistent in time and can be reduced to a 
single leaf through so-called \emph{cherry picking} operations, and have several 
applications, including the development of novel metric spaces on 
networks~\cite{cardona2024comparison,landry2023fixed} and allowing simple algorithms for 
isomorphism and network containment~\cite{janssen2021cherry}.

Before proceeding, for completeness we mention some further problems related to finding 
leaf IPPs. In~\cite{van2017fixed} 
the problem of removing arcs from a DAG so that every connected component contains exactly one leaf is considered. 
In addition, the problem of finding $k$ vertex-disjoint paths between specified start and 
end vertex pairs $(s_1, t_1), \ldots, (s_k, t_k)$, without necessarily covering 
all vertices, has received considerable attention in both the induced and non-induced settings.  
In DAGs, this is polynomial-time solvable if $k$ is fixed~\cite{fortune1980directed} (see also~\cite{tholey2012linear} 
for a linear-time algorithm when $k=2$), but the problem is NP-complete on DAGs already when $k=2$ 
if the paths are required to be induced~\cite{kawarabayashi2008induced}.  The induced 
paths version is polynomial-time solvable on directed planar graphs~\cite{kawarabayashi2008induced} 
for fixed $k$, and when $k$ is treated as a parameter, finding $k$ edge-disjoint paths is 
$W[1]$-hard~\cite{slivkins2010parameterized}, meaning that there is probably no 
algorithm with time complexity of the form $f(k) \cdot n^c$, for some function $f$ and constant $c$.
Also see~\cite{berczi2017directed,lopes2022relaxation} for other types of 
directed disjoint path problems and analogous results on undirected graphs.  

We now give a summary of the contents of the rest of the paper.  Section~\ref{sec:prelim} introduces the 
preliminary notions on phylogenetic networks and related structures, 
while Section~\ref{sec:forest-based} formally introduces forest-based networks and 
their variants, along with their correspondence with path partitions.
In Section~\ref{sec:hardness}, we show that it is NP-complete to partition a 
binary DAG into three induced paths, implying that the forest-based recognition problem 
is hard even on networks with three leaves.  Section~\ref{sec:tractable} focuses on 
tractable instances, where we show that the analogous problem is polynomial-time 
solvable on binary networks with two leaves, and on orchards.  In Section~\ref{sec:unrooted}, 
we introduce \emph{unrooted} forest-based networks, and show that they are hard to 
recognize even on networks with four leaves.  We conclude with a brief discussion and
presenting some open problems.

\section{Preliminaries}\label{sec:prelim}

In this section, we present some terminology for graphs that we will use in this paper,
most of which is standard in graph theory and phylogenetics (see e.g. \cite[Chapter 10]{steel2016phylogeny}).
For a positive integer $n$, we use the notation $[n] = \{1, 2, \ldots, n\}$.  
Let $N$ be a directed acyclic graph (DAG).
Denote its vertex-set by $V(N)$ and its arc-set by $A(N)$. 
If $(u, v) \in A(N)$ is an arc, then $u$ is an \emph{in-neighbor} of $v$ and $v$ an \emph{out-neighbor} of $u$.  The \emph{indegree} and \emph{outdegree} of a vertex are its number of in-neighbors and out-neighbors, respectively.  We say that $v \in V(N)$ is a \emph{root} of $N$ if $v$ has indegree $0$, and a \emph{leaf} of $N$ if it has outdegree $0$.  Note that roots and leaves are sometimes called sources and sinks, respectively.  We denote by $R(N)$ the set of roots of $N$ and by $L(N)$ its set of leaves. 
A vertex of $V(N) \setminus L(N)$ is called an \emph{internal vertex}. If an internal vertex $v$ has indegree at least $2$, then $v$ is called a 
\emph{reticulation}, and otherwise $v$ is a \emph{tree-vertex}.
A vertex of indegree $1$ and outdegree $1$ is called a \emph{subdivision vertex}.
A DAG is \emph{semi-binary} if every root has outdegree 2, every internal vertex has total degree 2 or 3, and every 
leaf vertex has indegree 1 or 2; it is called \emph{binary} if it is semi-binary and contains no subdivision vertex (that
is every internal vertex has total degree 3).

Suppose that $N$ is a DAG. Although this is standard notation, 
due to its importance we note that a \emph{directed path} in $N$ is a 
sequence of distinct vertices $v_1,v_2,\dots, v_k$, $k \ge 1$, in $V(N)$ such that
$(v_i,v_{i+1}) \in A(N)$ for all $i \in [k-1]$.  
If $(v_i, v_j) \notin A$ for any $j>i+1$ also holds, then the sequence forms an \emph{induced path}.
 Note that
abusing notation we shall sometimes also consider such a path or induced path
as just being its set of vertices $P = \{v_1, \ldots, v_k\}$, from which the ordering of the sequence can be inferred.
For $B \subseteq V(N)$, we write $N[B]$ 
for the directed subgraph of $N$ induced by $B$. 
We say that $N$ is \emph{connected} if the underlying undirected graph of $N$ is connected
(that is, there is an undirected path that connects any pair of its vertices),
and a connected component of $N$ is the subgraph of $N$ induced by the vertices of a connected component of its underlying undirected graph.
A \emph{tree} is a connected DAG with a single root and no reticulation vertex.  
A \emph{forest} is a DAG in which every connected component is a tree.

A \emph{path partition} of $N$ is a collection $\P$ of vertex-disjoint 
directed paths in $N$ whose union is $V(N)$ (using our notation, each element of $\P$ is a set of vertices); it 
is called an \emph{induced path partition} if every path in the
partition is an induced directed path in $N$. 
We may write PP for path partition and IPP for induced path partition.
A (induced) path partition is called
a \emph{leaf (induced) path partition} if every path in the partition 
ends in a leaf of $N$. 
Note that since the
paths partition $N$, it follows that in a leaf PP or a leaf IPP, each leaf of $N$ must be the end of some path. 
Observing that no two leaves are in the same path, a leaf PP or leaf IPP, if it exists, partitions $N$ into the smallest possible number of paths.
Also note that, as stated in the introduction, path partitions in 
graphs that contain a specified subset of vertices have been 
studied in the literature; see e.g. \cite[Section 3.2]{manuel2018revisiting}.


We say that a 
DAG $N$ with at least two vertices is a \emph{network} if: $N$ is connected; 
every root has outdegree at least $2$; 
every leaf has indegree $1$; 
every reticulation has outdegree $1$; 
and $N$ has no subdivision vertex.  
If $|V(N)| = 1$, then $N$ is a network  and $R(N) = L(N)$. 
If a $N$ is a tree, then $N$ is a \emph{phylogenetic tree}, and 
if it is a forest it is a \emph{phylogenetic forest}.
In addition if $N$ has a single root it is called a 
\emph{phylogenetic network}\footnote{In phylogenetics, it is common
to call $N$ a phylogenetic network \emph{on $X$}, where $X=L(N)$. But
since the labels of the leaves is not 
important in our arguments, we shall not follow this convention in this paper.}.

\section{Forest-based DAGs}\label{sec:forest-based}

In this section, we consider forest-based networks and some of their properties.
We call a DAG $N=(V,A)$  {\em weakly forest-based} if there exists $A' \subseteq A$ 
such that $F'=(V,A')$ is a forest with leaf set $L(N)$.
If in addition every arc in $A \setminus A'$ 
has its two endpoints in different trees of $F'$ (i.e. $F'$ is an induced forest in $N$),
then we call $N$ {\em forest-based}. 
If $N$ is (weakly) forest-based relative to some spanning forest $F$, 
we call $F$ a {\em subdivision forest (of $N$)}. These
definitions generalise the definition of a forest-based network
presented in \cite{huber2022forest}.

Forest-based networks
were first introduced in \cite{huber2022forest} as a generalization
of so-called \emph{overlaid species forests} \cite{huber2022overlaid}, and can 
be used to analyze an evolutionary process called \emph{introgression} (see
more about this evolutionary process in \cite{scholz2019osf}).
Note that as mentioned in the introduction 
forest-based phylogenetic networks are closely related to tree-based
phylogenetic networks. In particular, 
a binary phylogenetic network $N$ is \emph{tree-based} if it
contains a rooted spanning tree with leaf set $L(N)$ or,
equivalently, it admits a leaf path partition \cite[Theorem 2.1]{francis2018new}.
Thus, tree-based binary phylogenetic networks are weakly forest-based.

The following key result extends the above
stated relationships to DAGs. Its proof is almost identical 
to~\cite[Theorem 1]{huber2022forest}, but we include it for completeness.

\begin{theorem}\label{thm:fb-link}
Suppose that $N$ is a DAG. Then
\begin{itemize}
    \item[(i)] $N$ is weakly forest-based if 
and only if $N$ contains a leaf path partition.
\item[(ii)] $N$ is forest-based if 
and only if $N$ contains a leaf induced path partition.
\end{itemize}
\end{theorem}

\begin{proof}
We prove the result for forest-based DAGs; the proof is
the similar for weakly forest-based DAGs.

Suppose that $N$ admits a leaf induced path partition, then it is clearly
forest-based (with subdivision forest a collection of induced directed paths).
	
Conversely, suppose that $N=(V,A)$ is forest-based, with subdivision forest $F'$. If every 
connected component in $F'$ is a directed path, then
the converse holds as these paths must be induced and must each end at a leaf. So, suppose this is not the case, and that there exists a connected component 
$T'$ in $F'$ that is not a path. Then, as $T'$ is 
a tree, there must be a vertex $v$ with outdegree at least 2 such that no 
ancestor of $v$ in $T'$ has outdegree greater than 1. By removing 
all but one of the arcs from $T'$ with tail $v$, we obtain a new 
subdivision forest of $N$, which has more components than $F'$. Repeating 
this process if necessary, we eventually end up with a subdivision forest of $N$ that 
consists of a collection of induced directed paths. So $N$ admits a leaf induced path partition. 
\end{proof}

Using the link with path partitions
given in Theorem~\ref{thm:fb-link}, we shall show in the next section that it is NP-complete 
to decide whether a DAG $N$ is forest-based, 
even if $N$ is a binary, weakly forest-based network with three leaves.
In contrast, again using this link, we shall now explain why 
it is possible to decide whether a DAG $N$ is 
weakly forest-based in polynomial time in $|V(N)|$.
This is essentially proven in
\cite{francis2018new} for the special case that 
$N$ is a binary phylogenetic network, but we present the main ideas 
in the proof for DAGs for the reader's convenience.

For a DAG $N = (V, A)$, let $d(N)$ be the smallest number of vertex-disjoint 
paths that partition the vertex set of $N$. This number is closely
related to the size of a maximum matching in the following undirected bipartite graph $G(N)$
associated to $N$. The vertex bipartition of $G(N)$ is $\{V_1,V_2\}$, where
$V_1$ and $V_2$ are copies of $V$, and the edge set of $G(N)$ consists of
those $\{u,v\}$ with $u \in V_1$ and $v \in V_2$ such that there
is an arc $(u,v)$ in $A$. The proof of the following result is
more-or-less identical to that of \cite[Lemma~4.1]{francis2018new}\footnote{In
the statement of that lemma take $\mathcal N$ to be a DAG $N$ 
with leaf-set $X=L(N)$, $p(N)=d(N)-|X|$, and $u(\mathcal G_{\mathcal N})$ 
to be the number of unmatched vertices in $V_1$ relative to a maximum-sized matching of $G(N)$.}, 
and so we shall not repeat it here. Note that 
it can also be shown using \cite[Problem 26-2]{cormen2022introduction}, which 
yields essentially the same proof.
The main idea is that the matched vertices of $V_1$ can have their partner vertex from $V_2$ 
as their successor in a path, whereas the unmatched vertices consist 
of the ends of the paths.

\begin{theorem}\label{thm:fss}
Let $N$ be a DAG. Then $d(N)$ is equal to the 
number of unmatched vertices in $V_1$ relative to a maximum matching of $G(N)$.
\end{theorem}

The following corollary is the DAG-analogue of \cite[Corollary~4.2]{francis2018new};
the proof is essentially the same but we repeat it 
for the reader's convenience.

\begin{corollary}
Let $N$ be a DAG. Then $N$ is weakly forest-based if and only $G(N)$ has a matching
of size $|V(N)|-|L(N)|$. In particular, we can decide 
if $N$ is weakly forest-based in $O(|V|^{5/2})$ time.
\end{corollary}
\begin{proof}
The elements of $L(N)$ in $V_1$ can clearly never be matched, 
and so $G(N)$ has a matching of size $|V(N)|-|L(N)|$
if and only if $G(N)$ has a {\em maximum} matching of this size.
Now, by Theorem~\ref{thm:fss}, the latter holds if and only if $d(N)=|L(N)|$.
But this is the case if and only if $N$ is weakly forest-based.
The last statement follows since a matching in a bipartite graph
with $n$ vertices and $m$ edges can be found in $(m+n)\sqrt{n}$ time \cite{hopcroft1973n},
and $G(N)$ has $2|V(N)|$ vertices and 
$O(|V(N)|^2)$ edges since $G(N)$ has the same number of edges as $N$.
\end{proof}

\section{Hardness results}\label{sec:hardness}

In the main result in this section, we shall 
show that it is NP-complete to decide whether a binary network $N$ with 
three roots and three leaves is forest-based. We achieve this using the characterization from Theorem~\ref{thm:fb-link}. That is, we show that deciding whether $N$ admits a leaf IPP is NP-complete.  
Note that since $N$ has three leaves, this is equivalent to asking whether a given DAG can be partitioned into three induced paths.   Recall that the hardness results from~\cite{kawarabayashi2008induced} imply that it is NP-complete to find three vertex-disjoint induced paths with specified ends $(s_1, t_1), (s_2, t_2), (s_3, t_3)$.  Although these ends could be specified as the three roots and three leaves, we note that this problem does not reduce immediately to ours, because the latter has not been shown NP-complete on binary networks, and because we require our paths to cover every vertex.

The reduction that we shall use for the main result 
is from \textsc{Monotone NAE-3-SAT} \cite{dehghan2015complexity}.
In this problem, the input is a set of Boolean clauses, each containing exactly three positive literals (thus, no negation).  The goal is to find an assignment of the variables so that, for each clause, the variables of the clause are not all assigned true, and not all assigned false either.

\begin{theorem}\label{thm:three-paths-hard}
    It is NP-complete to decide whether a connected binary DAG with three roots and three leaves can be partitioned into three induced paths.

    Moreover, unless the ETH fails, under the same constraints the problem cannot be solved in time $2^{o(n + m)} n^c$, where $n, m$ are the number of vertices and arcs, respectively, and $c$ is any constant.
\end{theorem}

\begin{proof}
The problem is in NP since it is easy to verify that a given partition of the vertices of a DAG forms three induced paths.

For NP-hardness, 
consider an instance $\phi$ of the \textsc{Monotone NAE-3-SAT} problem, where $\phi$ has variables $x_1, \ldots, x_n$ and clauses $C_1, \ldots, C_m$, each with three positive literals.
We generate a connected binary DAG $N$ with three roots and three leaves as follows.
The main idea is that, in a desired leaf IPP consisting of three induced paths $P_1, P_2, P_3$, the first two paths   
$P_1$ and $P_2$ will first go through a set of vertices that represent a choice of values for the variables.  The vertices in $P_1$ will represent the variables assigned true, and the vertices in $P_2$ those assigned false.  A sequence of variable gadgets is introduced to enforce this.  After this, we introduce a gadget for each clause $C_j = (x_a \vee x_b \vee x_c)$ such that each of the three induced paths must go through a distinct vertex corresponding to $x_a, x_b, x_c$.  The paths $P_1, P_2$ will be able to ``pass through'' this gadget only if each of $P_1$ and $P_2$ has not encountered one of the $x_a, x_b, x_c$ vertices in the previous step (the third path $P_3$ is only there to cover the remaining vertex of the clause gadget).

For a variable $x_i$, $i \in [n]$, let $l(i)$ denote the number of clauses that contain $x_i$. 
 To ease notation below, we write $l := l(i)$, with the understanding 
 that $l$ depends on the variable $x_i$ under consideration.
Let $j_1, \ldots, j_{l}$ be the set of indices of the clauses that contain $x_i$, that is, 
$C_{j_1}, C_{j_2}, \ldots, C_{j_{l}}$ 
 is the set of clauses that contain $x_i$. 
Create a gadget $X_i$ that contains two  induced directed paths $X_i^1, X_i^2$ (see Figure~\ref{fig:XiYi}).   The directed path $X_i^1$ consists of $l + 2$ vertices $a_i \rightarrow a'_i \rightarrow x_i(j_1) \rightarrow x_i(j_2) \rightarrow \ldots \rightarrow x_i(j_l)$. 
The directed path $X_i^2$ consists of two vertices $b_i \rightarrow b'_i$.  
Then we add the arcs $(a_i, b'_i), (b_i, a'_i)$, which will allow switching paths.  
In addition, we connect the $X_i$ gadgets as follows (see Figure~\ref{fig:nae-full}).  
For each $i \in [n - 1]$, add the arc $(x_i(j_l), a_{i+1})$ 
and $(b'_i, b_{i+1})$.  

\begin{figure}[H]
    \centering
    \includegraphics[width=\textwidth]{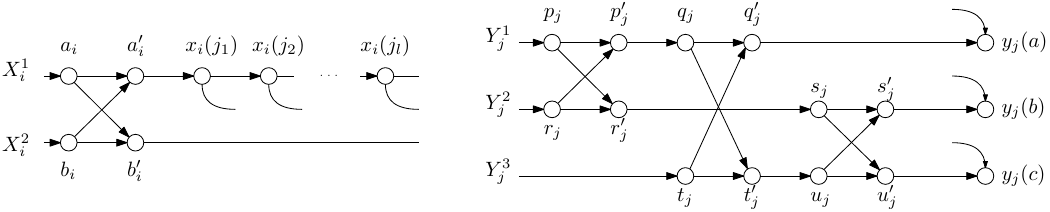}
    \caption{Left: one of the $X_i$ gadgets.  Here, $i > 1$ is assumed (if $i = 1$, $a_1$ and $b_1$ are roots).  Each vertex $x_i(j)$ has an out-neighbor $y_j(i)$ that is not shown. 
    Right: one of the $Y_j$ gadgets for a clause $C_j = (x_a \vee x_b \vee x_c)$.  The in-neighbors of $y_j(a), y_j(b), y_j(c)$ which are not shown are, respectively, $x_a(j), x_b(j), x_c(j)$.
    Note that the first vertex $t_1$ of $Y_1^3$ has no in-neighbor. }
    \label{fig:XiYi}
\end{figure}

Next, for each clause $C_j = (x_a \vee x_b \vee x_c)$, add a gadget $Y_j$ that consists of three induced directed paths $Y_j^1, Y_j^2, Y_j^3$ as in Figure~\ref{fig:XiYi}.  Roughly speaking, first there is a $Y_j^1 - Y_j^2$ path switcher, followed by a $Y_j^1 - Y_j^3$ path switcher, and then a $Y_j^2 - Y_j^3$ path switcher. The paths respectively end at vertices $y_j(a), y_j(b), y_j(c)$ with, respectively, additional in-neighbors $x_a(j), x_b(j), x_c(j)$.  
The switchers allow the permutation of the desired induced paths that enter the gadget 
in every possible way.
In more detail, the $Y_j$ gadget has the directed path $Y_j^1$ with 
vertices $p_j - p_j' - q_j - q_j' - y_j(a)$, the directed path $Y_j^2$ with vertices 
$r_j - r_j' - s_j - s'_j - y_j(b)$, and 
the directed path $Y_j^3$ with vertices $t_j - t'_j - u_j - u'_j - y_j(c)$.  We add the arcs $(p_j, r_j'), (r_j, p_j'), (q_j, t'_j), (t_j, q_j')$, and $(s_j, u'_j), (u_j, s_j')$. 

\begin{figure}[H]
    \centering
    \includegraphics[width=\textwidth]{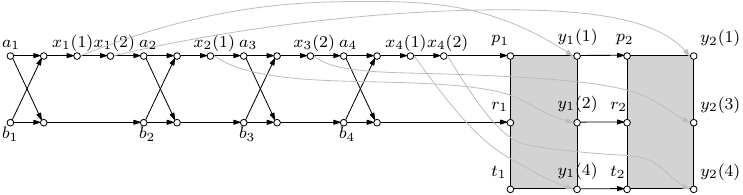}
    \caption{A detailed example over variables $x_1, x_2, x_3, x_4$ and clauses $C_1 = (x_1 \vee x_2 \vee x_4), C_2 = (x_1 \vee x_3 \vee x_4)$.  For clarity, only the vertices entering and exiting the $Y_j$ gadgets are shown.  As an example, notice that the vertex $x_3(2)$ exists because $x_3$ is present in $C_2$, which implies the presence of the arc $(x_3(2), y_2(3))$.}
    \label{fig:nae-full}
\end{figure}

To connect  the $C_j$ gadgets, for each $j \in [m - 1]$, add an arc from the last vertex of $Y_j^i$ to the first vertex of $Y_{j+1}^i$, for $i \in \{1,2,3\}$.  
Then add an arc from the last vertex of $X_n^i$ to the first vertex of $Y_1^i$, for $i \in \{1,2\}$. The first vertex $t_1$ of $Y_1^3$ has no in-neighbor and is therefore a root.  
Finally, noting that the vertices $x_a(j), x_b(j), x_c(j)$ exist, we 
also add the arcs $(x_a(j), y_j(a)), (x_b(j), y_j(b)), (x_c(j), y_j(c))$ (see 
Figure~\ref{fig:nae-full}).  

This completes the construction of $N$.
One can check that the network is binary.
We show that $\phi$ admits a not-all-equal assignment if and only $N$ admits a leaf IPP.

($\Rightarrow$)
Suppose that $\phi$ admits a not-all-equal assignment $A$, where we denote $A(x_i) \in \{T, F\}$ for the value of $x_i$ assigned by $A$.
The three induced paths of $N$ are constructed algorithmically.  The first phase corresponds to an assignment and puts the vertices of the $X_i$ gadgets, plus $p_1, r_1$, into the induced paths $P_1, P_2$ (and $t_1$ in $P_3$).  In a second phase, we extend those paths to include the vertices of the $Y_j$ gadgets.

In the first phase, we begin by initiating the construction of path $P_1$, which starts at $a_1$.  An illustration is provided in Figure~\ref{fig:nae-sol}.  The path is built iteratively for $i = 1,\dots,n$ in this order, with the invariant that before applying the $i$-th step, $P_1$ contains exactly one of $a_i$ or $b_i$ (which is true for $i = 1$).  
So, for $i \in [n - 1]$, assume that $P_1$ currently ends at $a_i$ or $b_i$.  If $A(x_i) = T$, $P_1$ goes to $a_i'$, then through the vertices $x_i(j_1) - \ldots - x_i(j_l)$, and then to $a_{i+1}$ (where here, $l = l(i)$).  If $i = n$, $P_1$ is extended in the same manner except that we go to $p_1$, the first vertex of $Y_1^1$, as shown in Figure~\ref{fig:nae-sol}.  If $A(x_i) = F$, then $P_1$ goes to $b'_i$ and then to $b_{i+1}$ (or, if $i = n$, to $r_1$, the first vertex of $Y_1^2$).

Note that at this stage, $P_1$ is an induced path, since the only vertices of the $X_i$ gadgets with two in-neighbors are the $a_i', b_i'$ vertices, and we cannot include both of their in-neighbors in the same path. \emph{Also note for later reference that for any vertex of the form $x_i(j)$, $P_1$ contains $x_i(j)$ if and only if $A(x_i) = T$.}

\begin{figure}[H]
    \centering
    \includegraphics[width=\textwidth]{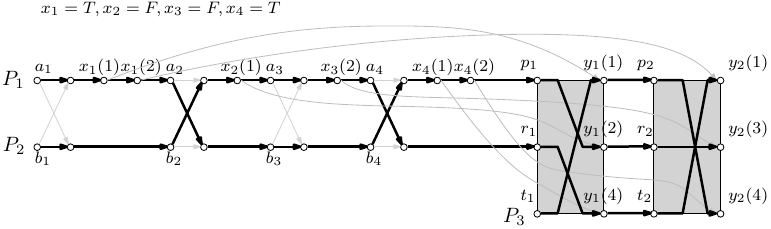}
    \caption{An induced path partition that corresponds to assigning $x_1, x_4$ to true and $x_2, x_3$ to false.  Notice that, for example, $P_1$ goes through $y_1(2)$ and $y_2(3)$ because it avoided going through $x_2(1)$ and $x_3(2)$.}
    \label{fig:nae-sol}
\end{figure}

Next, we let $P_2$ consist of all the vertices of the $X_i$ gadgets, for $i \in [n]$, that are not in $P_1$, plus $\{p_1, r_1\} \setminus P_1$.  In other words, to construct $P_2$ follow the same procedure as $P_1$, but start at $b_1$ and apply the opposite of the assignment $A$.  At this stage, $P_2$ is also induced by the same arguments.
\emph{Moreover, $P_2$ contains $x_i(j)$ if and only if $A(x_i) = F$.}

Finally, let $P_3$ consist of the vertex $t_1$, the first vertex of $Y_1^3$ (which is a root).  This completes the first phase.

In the second phase, we next extend, in an iterative manner, the induced paths constructed so far by adding the vertices of the $Y_j$ gadgets.  
For $j = 1,\dots, m$ in this order, assume that we have reached a point where the last vertices of $P_1, P_2, $ and $P_3$ are $p_j, r_j$, and $t_j$ (without assuming which vertex currently ends which path).  Note that this is true for $j = 1$ when we start this phase.
Let $C_j = (x_a \vee x_b \vee x_c)$.  
Since $A$ is a not-all-equal assignment,
one of the variables is false, say $x_d$ where $d \in \{a,b,c\}$, and one of the variables is true, say $x_e$ where $e \in \{a,b,c\}$ and $e \neq d$.

We now extend $P_1, P_2, P_3$ so that they cover all the vertices of the $Y_j$ gadget, and so that $P_1$ ends at $y_j(d)$ and $P_2$ ends at $y_j(e)$ (and $P_3$ ends at the remaining $y_j$ vertex).
This is always possible since 
the three path switchers in the gadget $Y_j$ can be used 
to extend and redirect the paths $P_1$ and $P_2$ to the desired exit vertex
(see Figure~\ref{fig:ipp-permut}).
For example, if $P_2, P_1, P_3$ enter at $p_j, r_j, t_j$, respectively, and we want to 
extend and redirect them to $y_j(b), y_j(c), y_j(a)$, respectively, then we would use the $123 \rightarrow 231$ permutation (see the caption of the figure).

\begin{figure}[H]
    \centering
    \includegraphics[width=\textwidth]{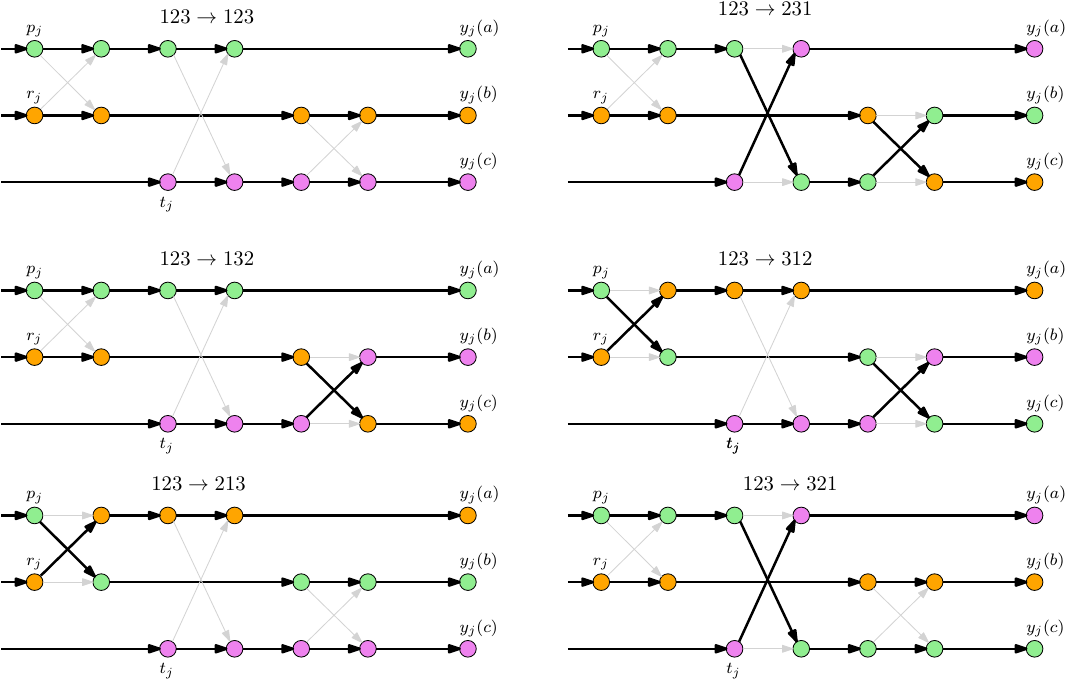}
    \caption{An illustration of how $P_1, P_2, P_3$ can be constructed to make them reach any set of desired ends of the $Y_j$ gadget.
    Vertices of the same color are in the same path, and the arcs in bold show the arcs of the three paths.  The numbers $1,2,3$ refer to the index of the entering path from top to bottom.  The permutation $123 \rightarrow ijk$ means that the first, second, and third paths exit as the $i$-th, $j$-th, and $k$-th paths, respectively.   
    }
    \label{fig:ipp-permut}
\end{figure}

Now, consider the vertices with two in-neighbors that are added to the paths at this stage.
These include  the $p'_j, q'_j, r'_j, s'_j, t'_j, u'_j$ vertices, which in all six cases of Figure~\ref{fig:ipp-permut} have exactly one in-neighbor in the path 
that contains them (one way to verify this is to check that those vertices always have exactly one in-neighbor of the same color).  Thus these vertices cannot create non-induced paths.

The other vertices with two in-neighbors
are $y_j(a), y_j(b), y_j(c)$. 
Since $P_3$ does not contain any $x_i(j)$ vertex, it remains induced.  
As for $P_1$, because $x_d$ is chosen as a false variable, $P_1$ does not contain $x_d(j)$.  Therefore, adding $y_j(d)$ to $P_1$ safely preserves the induced property.  
For $P_2$, because $x_e$ is assigned true, $P_2$ does not contain $x_e(j)$ and $y_j(e)$ can be added to $P_2$.  It follows that extending the paths to cover the $Y_j$ vertices preserve the induced property of each path.

Now, if $j < n$, we then add to each path the single out-neighbor of their respective $y_j(a), y_j(b), y_j(c)$ vertices (which are $p_{j+1}, r_{j+1}, t_{j+1}$, vertices of in-degree $1$ that cannot create non-induced paths by appending them), which ensures that the extension can be 
applied to the next gadget.  When we reach $j = m$, the paths end at the leaves of the $C_m$ gadget, which 
concludes the proof that the three induced paths can be constructed as desired.

($\Leftarrow$)
Suppose that $N$ can be partitioned into three induced paths $P_1, P_2, P_3$, where $P_1$ starts at $a_1$, $P_2$ starts at $b_1$, and $P_3$ starts at $t_1$ in the $Y_1$ gadget.  
Let $i \in [n]$, with $l := l(i)$.  Observe that in the $X_i$ gadget, each $x_i(j)$ vertex has a single incoming arc.  Because each vertex aside from the roots must have an in-neighbor in its path, each of these arcs must be in the same induced path.  Moreover, the root of $P_3$ cannot reach these vertices, and therefore $x_i(j_1), \ldots, x_i(j_l)$ are either all in $P_1$, or all in $P_2$.  Also note that by this argument, none of the arcs $(x_i(j), y_j(i))$ can be 
contained in $P_1$ or $P_2$, because the out-neighbor of each $x_i(j)$ in its path must be the vertex other than $y_j(i)$ (one can check that this is also true for the $x_i(j_l)$ vertex, the last vertex of $X_i^1$).  In other words, $x_i(j)$ and $y_j(i)$ cannot be in the same path.
 
Now, consider the assignment $A$ that, for each $i \in [n]$, puts $A(x_i) = T$ if and only if $x_i(j_1), \ldots, x_i(j_l) \in P_1$, where again for each $x_i$, $l = l(i)$ is the number of clauses containing $x_i$.  
As argued above, all the $x_i(j)$'s are in the same path, so $A(x_i)$ is a well-defined assignment.  
We argue that $A$ is a not-all-equal assignment of $\phi$.  Let $C_j = (x_a \vee x_b \vee x_c)$ be a clause.  
Observe that in $N$, none of the vertices in $y_j(a), y_j(b), y_j(c)$ reach each other.  Therefore, they must all be in distinct induced paths.
In particular, $P_1$ must go through one of those, say $y_j(d)$, where $d \in \{a,b,c\}$.  This means that $P_1$ cannot contain $x_d(j)$, as otherwise $P_1$ would not be induced (since $(x_d(j), y_j(d))$ exists but it is not used by $P_1$).  
This means that we assign $A(x_d) = F$, and thus at least one variable of $C_j$ is false.
Also, $P_2$ must go through one of the three vertices as well, say $y_j(e)$ where $e \in \{a,b,c\}$ and $e \neq d$.  
As before, this means that $P_2$ does not contain $x_e(j)$, and thus $P_1$ must contain $x_e(j)$.  We assign $A(x_e) = T$, and thus at least one variable of $C_j$ is true.  
Since this holds for every $C_j$, $A$ is a not-all-equal assignment of $\phi$.

We have thus shown NP-completeness.  As for the ETH lower bound, it was shown in~\cite[Proposition 5.1]{antony2024switching} that, unless the ETH fails, \textsc{Monotone NAE-3-SAT} cannot be solved in time $2^{o(n + m)} n^c$, where $n$ is the number of variables and $m$ the number of clauses.  Consider the number of vertices and arcs of a constructed instance $N$.  The number of vertices of each $X_i$ gadget is $4 + l(i)$ and, since each clause has three variables, the total number of vertices in the $X_i$ gadgets is $4n + \sum_{i=1}^n l(i) = 4n + 3m$.
Each $Y_i$ gadget has $15$ vertices and 
the total number of vertices in the $Y_i$ gadgets is $15m$.  Therefore, $V(N) \in O(n + m)$ and, since $N$ is binary, $A(N) \in O(n + m)$.  It follows that a $2^{o(|V(N)| + |A(N)|)} n^c$ time algorithm for the forest-based recognition problem could be used to solve \textsc{Monotone NAE 3-SAT} in time $2^{o(n + m)} n^c$, which cannot occur if the ETH is true.
\end{proof}

It may be interesting to note that in the reduction of Theorem~\ref{thm:three-paths-hard}, only two paths are ``useful'', in the sense that they respectively correspond to the variables assigned positively and negatively.  The third path is more of a ``dummy'' path solely used to cover unused vertices.  This may lead to the intuition that the problem is NP-complete on two paths, but our positive result in the next section shows that the dummy path is necessary to make the problem hard, at least in the binary case.

The last theorem answers a question in Fernau et al.~\cite[Section 9]{fernau2023parameterizing}, in which it was asked whether partitioning a DAG into at most $k$ induced paths is in XP, i.e.,
whether it can be done in polynomial time if $k$ is fixed.  Recall that a problem is para-NP-hard with respect to a parameter $k$ if the problem is NP-hard even when $k$ is a fixed constant, making it unlikely to belong to the XP complexity class.

\begin{corollary}
    The problem of partitioning a connected binary DAG into at most $k$ induced paths is NP-complete for every fixed $k \geq 3$.  The problem is therefore para-NP-hard with respect to parameter $k$.
\end{corollary}

\begin{proof}
    Theorem~\ref{thm:three-paths-hard} shows that the problem is NP-complete for $k = 3$.  For $k > 3$, we can easily reduce from the case of partitioning a connected binary DAG into three induced paths as follows.  
    Given an instance $N$ of the latter, obtain $N'$ by adding to $N$ a connected component consisting of any binary tree $T$ with $k - 3$ leaves (with arcs directed away from the root).  Then take any root $r$ of $N$, add a new vertex $v$, and give to $v$ as out-neighbors $r$ and the root of $T$.  
    This resulting $N'$ is a connected binary DAG. 
    If $N$ can be partitioned into three induced paths, we take that partition and add any induced path partition of $T$ into $k - 3$ induced paths (which is easily seen to exist, since trees are forest-based), and add $v$ to the path that contains the root of $T$.  Conversely, if $N'$ can be partitioned into $k$ induced paths, then $k - 3$ of these paths must partition $T$, because leaves are in distinct paths and the vertices of $T$ only reach those leaves.  This means that the vertices of $N$ must be partitioned into the remaining three paths (possibly with $v$, which we may delete from its path).
\end{proof}

With a slight adaptation of the above, we can also show that 
even tree-based, binary phylogenetic networks are no easier 
to deal with than binary DAGs with three roots.

\begin{corollary}
    It is NP-complete to decide whether a tree-based, binary phylogenetic network with three leaves is forest-based.
\end{corollary}

\begin{proof}
    The NP membership is as in Theorem~\ref{thm:three-paths-hard}.  Let us argue that NP-hardness still holds even if we require $N$  to satisfy all the requirements of a tree-based, binary phylogenetic network. 

    Let $N$ be an instance of the forest-based recognition problem produced by the reduction of Theorem~\ref{thm:three-paths-hard}, where $N$ is binary and has three roots $r_1, r_2, r_3$ and three leaves $\ell_1, \ell_2, \ell_3$. 
    It can be seen from Figure~\ref{fig:nae-full} that $N$ is weakly forest-based.  Indeed, $V(N)$ can be partitioned into three (non-induced) paths as follows: one path concatenates all the top paths of the $X_i$ gadgets, followed by all the top paths of the $Y_j$ gadgets; one path concatenates all the bottom paths of the $X_i$ gadgets, followed by all the middle paths of the $Y_j$ gadgets; one path concatenates all the bottom paths of the $Y_j$ gadgets.
    Also notice that $N$ is connected, all roots have outdegree $2$, all reticulations have outdegree $1$, and $N$ has no subdivision vertex.  Thus only the requirement on leaves having indegree $1$ is missing to argue that $N$ is a network.
    This can easily be dealt with by creating a new network $N'$, obtained taking the leaves of $N$ and, for each leaf $v$ of indegree $2$, adding a new leaf $v'$ whose single in-neighbor is $v$.  Then, $N$ can be split into three induced paths if and only if $N'$ can, since paths of $N$ that end at a leaf $v$ can be extended with the new leaf $v'$, and conversely for paths of $N'$ that end at such a leaf $v'$, it suffices to remove it.
    Also note that the aforementioned path partition of $N$ can easily be extended to incorporate $v'$ in the same manner.
    It follows that the problem is hard on binary, weakly forest-based \emph{networks} with three roots and three leaves.

    To argue that the problem is also hard on binary phylogenetic networks, that
    is, binary networks with a single root, take $N'$ and 
    obtain $N''$ by adding two vertices $r, r'$, where $r$ has out-neighbors $r', r_3$ and $r'$ has out-neighbors $r_1, r_2$.  Note that $N''$ is still a binary network and is single-rooted.  Moreover, $N'$ can be split into three induced paths if and only if $N'$ can be as well.  Indeed, if $\{P_1, P_2, P_3\}$ is such a partition for $N'$, where $P_1$ starts at $r_1$, then we can add the sub-path $r \rightarrow r' \rightarrow r_1$ at the start of $P_1$, which results in a leaf IPP for $N''$.  Conversely, if $\{ P_1', P_2', P_3'\}$ partitions $N''$ into induced paths, the $r_i$ vertices must be in distinct paths are can only be preceded by $r$ or $r'$.  Therefore, by removing $r$ and $r'$ from these paths, we obtain a leaf IPP for $N'$. One can also see that $N''$ is tree-based as follows: take the subgraph of $N'$ consisting of the three paths from the above path partition, each starting at a distinct root, then add $r, r'$, and their incident arcs to this subgraph.  This results in a spanning tree of $N''$ whose leaves are $L(N'')$, which shows that $N''$ is tree-based.
    Therefore, the problem is also hard on tree-based, binary phylogenetic networks with three leaves.
\end{proof}

\section{Two tractable cases}\label{sec:tractable}

In this section, we first show that the leaf IPP problem 
is polynomial-time solvable on semi-binary DAGs with two leaves, showing that the hardness result from the previous is, in some sense, tight.  Note that the positive result also holds on binary networks, in particular.

We then show that the class of networks known as \emph{orchards} are all forest-based, as they always admit a leaf IPP.  
This generalizes~\cite[Theorem 2]{huber2022forest}, 
in which it is shown that binary \emph{tree-child} networks are forest-based, where 
a network is tree-child if all of its internal vertices of have a child that is a tree-vertex.

\subsection{Partitioning semi-binary DAGs into two induced paths}

In the following, we shall assume that $N$ is a semi-binary DAG 
that we want to partition into two induced paths.
Unlike in the previous section, we do not assume that the roots and leaves of the desired paths are specified, and so we first study a slightly different variant of the problem.

Given a semi-binary DAG $N$ and four distinct vertices $s_1, s_2, t_1, t_2$ of $N$, we ask: can the vertices of $N$ be partitioned into two induced paths $P_1, P_2$, such that the paths start at $s_1$ and $s_2$, and end at $t_1$ and $t_2$?  Note that the given vertices are not required to be roots or leaves, and that the path that starts with $s_1$ could end at either $t_1$ or $t_2$.
We call this the \textsc{Restricted 2-IPP} problem.  We then discuss how this can be used to solve the general problem.
Again, note that finding two disjoint induced paths between specified pairs $(s_1, t_1), (s_2, t_2)$ is NP-complete on DAGs~\cite{kawarabayashi2008induced}, but that the problem differs from ours since we must cover every vertex and restrict the problem to binary networks.  In fact, the latter two requirements are needed for our algorithm to be correct.

We reduce the \textsc{Restricted 2-IPP} problem to 2-SAT, which given a set of Boolean clauses with two literals each, asks whether there is an assignment that satisfies all clauses.  
For our purposes, it is sufficient to express our 2-SAT instances as constraints of the form $(x = y)$ or $(x \neq y)$, where $x$ and $y$ are literals (i.e., $x, y$ are variables or their negation), and where these constraints require the literals to be either equal or distinct, respectively.  In 2-SAT, $(x = y)$ is equivalent to having the clauses $(\neg x \vee y)$ and $(x \vee \neg y)$, and $(x \neq y)$ is equivalent to having the two clauses $(x \vee y)$ and $(\neg x \vee \neg y)$.

Given a semi-binary DAG $N$ and four vertices $s_1, s_2, t_1, t_2$, we create a Boolean variable $x_v$ for each $v \in V(N)$.  The variable $x_v$ is interpreted to be $true$ when $v$ belongs to $P_1$, and $false$ when $v$ belongs to $P_2$.  Using this variable representation, the goal is to assign each vertex to a path while satisfying all requirements of leaf IPPs.  Our 2-SAT instance is then obtained by adding the following set of constraints:

\begin{enumerate}
\item 
\emph{leaves and roots are in distinct paths}: add the constraints $(x_{s_1} \neq x_{s_2})$ and $(x_{t_1} \neq x_{t_2})$.

\item 
\emph{roots are roots, leaves are leaves}: for $i \in \{1,2\}$, and for each in-neighbor $w$ of $s_i$, add the constraint $(x_w \neq x_{s_i})$.  Then for each out-neighbor $w$ of $t_i$, add the constraint $(x_w \neq x_{t_i})$.

\item 
\emph{forced successors}: let $v \neq t_1, t_2$ be a vertex of $N$ with a single out-neighbor $w$.  Add the constraint $(x_v = x_w)$.
\item 
\emph{exactly one successor}: let $v \neq t_1, t_2$ be a vertex with two out-neighbors $u, w$.  Add the constraint $(x_u \neq x_w)$.  

\item 
\emph{exactly one predecessor}: let $v \neq s_1, s_2$ be a vertex with two in-neighbors $u, w$.  Add the constraint $(x_u \neq x_w)$.
\end{enumerate}
Note that we have not modeled the constraint that vertices with a single parent should be forced to be equal, since this is implied by the other constraints.  We show that this reduction is correct and leads to a polynomial time algorithm.

\begin{theorem}
    The \textsc{Restricted 2-IPP} problem can be solved in time $O(|V(N)|)$ on a semi-binary DAG $N$.
\end{theorem}

\vspace{-2mm}

\begin{proof}
    Let $N$ be a semi-binary DAG and $s_1, s_2, t_1, t_2$ be the four given vertices.
    Note that if some vertex $v \neq s_1, s_2$ is a root of $N$, then no IPP with two paths can start with $s_1, s_2$.  Likewise, if $v \neq t_1, t_2$ is a leaf of $N$, no solution is possible.  If one such case arises, we reject the instance, so from now on we assume that $N$ has no roots or leaves other than $s_1, s_2$ or $t_1, t_2$, respectively.
    We next show that our reduction to 2-SAT is correct.
    
    Suppose that $N$ can be partitioned into two induced paths $P_1, P_2$ whose roots are $s_1, s_2$ and whose leaves are $t_1, t_2$.  For each $v \in N(V)$, assign $x_v = true$ if $v \in P_1$, and $x_v = false$ if $v \in P_2$.  We argue that each constraint is satisfied.

    Because $s_1, s_2$ are in different paths, $x_{s_1} \neq x_{s_2}$ holds.  For similar reasons, $x_{t_1} \neq x_{t_2}$ also holds.  
    Moreover, for $i \in \{1,2\}$, as $s_i$ is the start of one of the induced paths, no in-neighbor $w$ of $s_i$ is in the same path as $s_i$.  Therefore, $x_w \neq x_{s_i}$.  Similarly, since $t_i$ has no out-neighbor $w$ in its path, $x_w \neq x_{t_i}$. 
    
    Let $v$ be a vertex other than $t_1, t_2$ with a single out-neighbor $w$.  Since $v$ must have a successor in its path, $v$ and $w$ must be in the same path and thus $x_v = x_w$, thereby satisfying the forced successor constraint.  
    Suppose that $v$ has two out-neighbors $u, w$.  Since $v \neq t_1, t_2$, it has some out-neighbor in its path, and in fact exactly one out-neighbor since the paths are induced.  It follows that $u$ and $w$ are in distinct paths and $x_u \neq x_w$.
    Finally, suppose that $v \neq s_1, s_2$ has two in-neighbors $u, w$.  Exactly one of them must be in the same path as $v$ (not both, because of the induced property), 
    and so again $x_u \neq x_w$.  We deduce that our assignment satisfies our 2-SAT instance.

Conversely, suppose that some assignment of the $x_v$ variables satisfies the 2-SAT instance. 
We claim that $P_1 = \{ v : x_v = true \}$ and $P_2 = \{ v : x_v = false \}$ form an induced path partition of $N$.  
These sets clearly partition $V(N)$.   Note that $x_{s_1} \neq x_{s_2}$ implies that $s_1$ is in one path and $s_2$ in the other.  Without loss of generality, we assume that $s_1 \in P_1, s_2 \in P_2$.  Also note that $x_{t_1} \neq x_{t_2}$ implies that $t_1, t_2$ are in different paths, although we do not assume which is in which.  
 Let $t_i$ be the vertex in $P_1$.  We argue that $P_1$ is an induced path that starts at $s_1$ and ends at $t_i$ (the proof for $P_2$ is identical).

First note that because $x_w \neq x_{s_1}$ for every in-neighbor $w$ of $s_1$, no such in-neighbor is in $P_1$.  Likewise, $t_i$ has no out-neighbor in $P_1$ because of the constraints $x_w \neq x_{t_i}$.
Let $v \in P_1 \setminus \{t_i\}$.  Note that because we initially checked that only $t_1, t_2$ could be leaves of $N$, $v$ is not a leaf of $N$.
If $v$ has a single out-neighbor $w$ in $N$, then $x_v = x_w$ and $w$ is also in $P_1$.
If $v$ has two out-neighbors $u, w$, because $x_u \neq x_w$, exactly one of $x_u, x_w$ is $true$ and is in $P_1$.  Thus, every vertex in $N[P_1]$ has a single out-neighbor, except $t_i$ which has no out-neighbor.  

Next let $v \in P_1 \setminus \{s_1\}$.  
If $v$ has two in-neighbors $u, w$ in $N$, exactly one of them is in $P_1$ because of $x_u \neq x_w$.  It follows that each vertex of $P_1$ has at most one in-neighbor in $N[P_1]$, except $s_1$.

Because in $N[P_1]$, every vertex has in-degree and out-degree at most $1$, and because $N$ is acyclic, $N[P_1]$ is a collection of paths.  There can only be one such path because, as we argued, every vertex except $t_i$ has an out-neighbor in $P_1$.  Moreover, $P_1$ is induced because none of its vertices has two out-neighbors in $P_1$.  

By the same arguments, $P_2$ induces path, and therefore $\{P_1, P_2\}$ is a partition of $N$ into two induced paths, such that $P_1$ starts at $s_1$ and ends at $t_i$, and $P_2$ starts at $s_2$ and ends at the other $t_j$ vertex.

It only remains to justify the complexity.  Our 2-SAT instance contains $O(|V(N)|)$ variables and clauses, since each vertex generates $O(1)$ clauses.  Then, we can use a linear-time algorithm~\cite{aspvall1979linear} to solve the 2-SAT instance.
\end{proof}

If $s_1, s_2, t_1, t_2$ are not known in advance, we can simply guess them, which leads to the following.

\begin{corollary}
    Let $N$ be a semi-binary DAG.  Then we can decide whether $N$ can be partitioned into two induced paths in time $O(|V(N)|^3)$. 

    Moreover, if $N$ has two leaves, we can decide whether $N$ admits a leaf IPP in time $O(|V(N)|^2)$.
\end{corollary}

\begin{proof}
    We may assume that $N$ has at most two roots and at most two leaves, otherwise no IPP with two paths is possible.
    Since $N$ is a DAG, it has at least one root $s_1$ and one leaf $t_1$, which must start and end some path.  If $N$ has another root $s_2$ and another leaf $t_2$, they must also start and end a path, and it suffices to run our algorithm for \textsc{Restricted 2-IPP} on the four vertices.

    If $N$ does not have another root but has another leaf $t_2$, we iterate over every vertex that we label as $s_2$ and, for each such vertex, we run our algorithm for \textsc{Restricted 2-IPP}.  If $N$ can be partitioned into two induced paths, there exists a value of $s_2$ on which the algorithm returns a positive answer and we will find it.  This solves the case where $N$ has two leaves in time $O(|V(N)|^2)$.  

    The same complexity can be achieved if $N$ has another root but no other leaf, by iterating over every possible $t_2$.  If $N$ has only one root and one leaf, we iterate over all the $O(|V(N)|^2)$ combinations of $s_2, t_2$ and run our algorithm for \textsc{Restricted 2-IPP}, for a total time of $O(|V(N)|^3)$.  
\end{proof}

\subsection{Orchard networks}

As mentioned in the introduction, tree-based phylogenetic networks were first
introduced as phylogenetic networks that can be obtained from a 
rooted tree $T$ by adding some arcs between some 
of the vertices of $T$ \cite{francis2015phylogenetic} 
(in the terminology introduced above $T$ is a subdivision-tree for $N$).
As we have seen, it is NP-complete to decide if a 
binary, tree-based phylogenetic network is forest-based.
The main difficulty in recognizing when a 
tree-based phylogenetic network is forest-based 
occurs when some of these extra arcs are between ancestors and 
descendants in the subdivision-tree.  Indeed, if every extra arc is 
between incomparable vertices of the subdivision-tree, then it 
is easy to partition the subdivision-tree into 
induced paths while ignoring these extra arcs.  

This suggests that tree-based phylogenetic networks 
with ``time-consistent lateral arcs'' should be forest-based. 
Interestingly, such phylogenetic networks 
are precisely defined in~\cite{van2022orchard}, where it is  
shown that they correspond to a special class of phylogenetic networks
called \emph{orchard networks} \cite{erdHos2019class}.
The authors in~\cite{van2022orchard} also show that, by 
allowing non-binary orchard phylogenetic networks, 
one obtains a class of networks that is strictly 
broader than time-consistent tree-based networks.  
We now extend orchard networks even further to the 
DAG setting, and show that all such networks are forest-based.

Let $N$ be a DAG with no subdivision vertex in which all leaves have in-degree $1$ (with $N$ not necessarily binary, single-rooted, nor connected).  A \emph{cherry} of $N$ is a pair of distinct leaves $(x, y)$ such that, if $x'$ and $y'$ are the respective in-neighbors of $x$ and $y$, 
either $x' = y'$, or $y'$ is a reticulation and $(x', y') \in A(N)$.  When $x' = y'$, $(x, y)$ is called a \emph{standard cherry}, and in the second case $(x, y)$ is called a \emph{reticulated cherry}.
The \emph{cherry-picking operation} on cherry $(x, y)$ transforms $N$ as follows:
if $(x, y)$ is a standard cherry, remove $y$ and its incident arc, and suppress the possible resulting subdivision vertex; 
if $(x, y)$ is a reticulated cherry, remove the arc $(x', y')$ and suppress the possible resulting subdivision vertices. In case $N$ is a binary phylogenetic network, this 
definition agrees with the operation proposed for the original orchard networks in \cite[p.35]{erdHos2019class}.

\begin{figure}
    \centering
    \includegraphics[width=.9\textwidth]{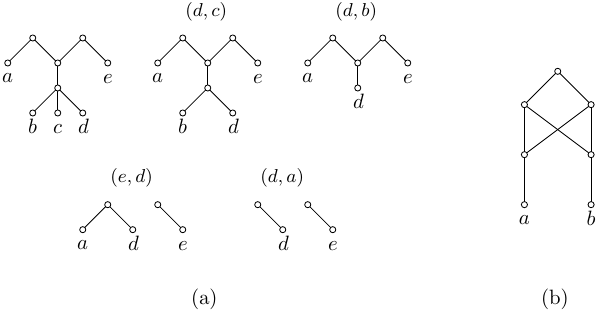}
    \caption{(a) A network $N$ reduced by a sequence of four cherry-picking operations.  The pairs on top indicate the operations performed to obtain the network (all arcs point downwards).
    (b) A forest-based network that is not an orchard.}
    \label{fig:cherrypick}
\end{figure}

A DAG is \emph{reduced} if each of its connected components has only one arc, whose endpoints are a root and a leaf.  
A DAG $N$ is \emph{reducible} if there exists a sequence of cherry-picking operations that can be applied to $N$ to transform it into a reduced DAG.  
We say that such a sequence \emph{reduces} $N$.
A network $N$ is a \emph{orchard} if there is a sequence of cherry-picking operations that reduces it.   
See Figure~\ref{fig:cherrypick}.a for an example with two roots.  Also, note that not every forest-based network is 
orchard.  The network in Figure~\ref{fig:cherrypick}.b is not an orchard network since 
it contains no cherry, but it clearly admits a leaf induced path partition.

\begin{theorem}
All reducible DAGs admit a leaf induced path partition.  Consequently, 
all orchards are forest-based.
\end{theorem}

\begin{proof}
	We use induction on the number of cherry-picking operations needed to reduce $N$.  If $N$ can be reduced with $0$ operations, then every connected component is a path with two vertices and $N$ trivially admits a leaf induced path partition.  
	Assume that $N$ requires at least one operation to be reduced and that the statement holds for DAGs that require less.  Consider a minimum-length sequence of cherry-picking operations that reduces $N$, and let $(x, y)$ be the first cherry in this sequence.  Let $N'$ be the DAG obtained from $N$ after picking $(x, y)$. 
	Note that $N'$ is reducible in one less operation than $N$.  Therefore, the induction hypothesis can be applied to $N'$ and we may thus assume that it admits a leaf induced path partition $\P'$.  We modify $\P'$ to obtain a 
 leaf induced path partition $\P$ of $N$.
	
	Suppose that $(x, y)$ is a standard cherry of $N$ and let $w$ be the common in-neighbor of $x$ and $y$.  If, after the removal of $y$, $w$ is not a subdivision vertex, then $N'$ has the same vertices as $N$, except $y$ which was removed.  In this case, we take $\P'$ and add the path consisting of $y$ by itself, which partitions $N$ into induced paths as desired.
	
 Otherwise, assume that $w$ is removed from $N'$ because it becomes a subdivision vertex.  This happens only if $w$ has $x$ and $y$ as out-neighbors, and only one in-neighbor $z$.  In $N'$, $z$ has become the in-neighbor of $x$.  Let $P_z \in \P'$ be the path that contains $z$.  
	We claim that we can assume that $P_z$ also contains $x$.
	If $P_z$ does not contain $x$, then $\{x\}$ by itself is a path of $\P'$ since $z$ is its sole in-neighbor.  In this case, let $P_1$ be the subpath of $P_z$ from its first vertex up until $z$, and let $P_2$ be the rest of the $P_z$ path.  In $\P'$, we can replace the two paths $P_z, \{x\}$ with $P_1 \cup \{x\}$, $P_2$, which are easily seen to be induced paths that cover the same vertices.  
	So we assume that $P_z$ uses the arc $(z, x)$.  
 
	Let us now revert the cherry-picking operation $(x, y)$ to go from $N'$ to $N$ by first subdividing $(z, x)$, thereby reinserting $w$ as a subdivision vertex.  
	By replacing $(z, x)$ in $P_z$ by the subpath $z - w - x$, we obtain a perfect induced path partition of the resulting network (since any path other than $P_z$ is unaffected by the subdivision, and because adding $w$ to $P_z$ preserves the induced property as $w$ has a single in-neighbor).  Then, reincorporate $y$ and the arc $(w, y)$.  Any path at this point is still induced, and it suffices to add $\{y\}$ by itself to obtain a leaf induced path partition of $N$.
	
	Suppose that $(x, y)$ is a reticulated cherry, with $x', y'$ the respective parents of $x, y$ and $y'$ a reticulation with $x'$ as an in-neighbor.  
	Let $p$ be the in-neighbor of $x$ in $N'$ and $q$ the in-neighbor of $y$ in $N'$.  Note that $p$ is either equal to $x'$, or $p$ is the in-neighbor of $x'$ in $N$, depending on whether $x'$ was suppressed as a subdivision vertex or not.  The same holds for $q$ and $y'$.  Let $P$ and $Q$  be the paths of $\P'$ that contain $p$ and $q$, respectively.  As before, we claim that we may assume that $x$ is in $P$ and $y$ in $Q$.  
	Indeed, if $x$ is not in $P$, then $x$ is a path by itself in $\P'$.  We can split $P$ in two such that the first subpath ends at $p$, and add $x$ as the out-neighbor of $p$ just as in the previous case.  After performing this replacement if needed, we assume that the arc $(p, x)$ is used by some path, and we can use the same argument to split $Q$ if needed and assume that $y$ is in $Q$ (noticing that applying this will not remove $(p, x)$).  Thus, $(p, x)$ is used by $P$ and $(q, y)$ is used by $Q$.
	
	To obtain a leaf induced path partition $\P$ of $N$, let us reverse the cherry-picking operation from $N'$ to $N$ one step at a time.  
	If $p \neq x'$, first subdivide $(p, x)$ to reincorporate $x'$, and in $P$ replace the arc $(p, x)$ with the subpath $p - x' - x$.  As before, this yields a perfect induced path partition of the resulting network.  If $p = x'$, then leave $P$ intact.  Likewise, if $q \neq y'$, subdivide $(q, y)$ to reincorporate $y'$ and in $Q$ replace $(q, y)$ with $q - y' - y$.  If $q = y'$ leave $Q$ intact.  Let $\P$ be the resulting leaf induced path partition.
	Finally, reinsert the arc $(x', y')$, which results in $N$ (and leave $\P$ unaltered).  If $\P$ now contains a non-induced path, it is because of the arc $(x', y')$, which is a problem only if $x'$ and $y'$ were in the same path.  If this were the case, that path in the previous network would reach $x'$ first then go to $y'$ or vice-versa, which we know does not occur because in their respective paths, the out-neighbor of $x'$ is $x$ and the out-neighbor of $y'$ is $y$.  It follows that $\P$ is a leaf induced path partition of $N$.
\end{proof}

Note that there
are several other well-studied classes of phylogenetic networks (see e.g. \cite{kong2022classes}).
In~\cite{huber2022forest}, the authors established most of the containment relationships of these classes with forest-based networks.  However, the computational complexity of the forest-based recognition problem remains open for several of these classes.  We have shown that the problem is hard on tree-based networks and easy for orchards, but we do not know whether the problem is NP-complete on other classes of interest.  This includes for instance tree-sibling networks, in which every reticulation has a sibling that is a tree-vertex, where a sibling is a vertex with the same parent (such a sibling may help redirecting partially constructed paths that cannot use the reticulation).
Other examples use the notion of \emph{visible} vertices, where a vertex $v$ is visible if there is a leaf such that $v$ is on every path from the root to that leaf.
In tree-child networks, every vertex is visible, and relaxing this condition yields classes on which finding leaf IPPs may be tractable.
One such class consists of reticulation-visible networks, in which every reticulation is visible, and another consists of nearly stable networks, where for each vertex $v$, either $v$ itself is visible, or its parents are visible.  
The complexity on these classes is open even for single-rooted binary networks.

\section{Unrooted forest-based networks}
\label{sec:unrooted}

In this section, we introduce an undirected analogue of
forest-based networks and 
consider some of their properties as compared with their rooted counter-parts. 
Most of the terms that we use for
undirected graphs are standard and similar to those 
used for directed graphs and so we shall not present 
definitions unless we think that clarification could be helpful.

A \emph{leaf} in an undirected graph is a vertex with degree 1.
An {\em unrooted phylogenetic network} is a (simple), connected, undirected graph 
$N=(V,E)$ with non-empty leaf-set $L(N)$, 
and that contains no vertices with degree two \cite{francis2018tree,hendriksen2018tree}.
The network $N$ is {\em binary} if every vertex in $V$ has degree 1 or 3, and
it is {\em tree-based} if it contains a spanning tree with leaf set $L(N)$.
Note that in contrast to
the rooted case, it is NP-complete to decide if a binary unrooted phylogenetic network
is tree-based \cite[Theorem 2]{francis2018tree}.

We now introduce the concept of forest-based unrooted networks.
In analogy with the rooted case, 
we call an unrooted phylogenetic network $N=(V,E)$ {\em forest-based}
if it contains a spanning forest $F$ with leaf set $L(N)$, such 
that every edge in $E \setminus F$ has its ends contained in different connected components of $F$, i.e.,
each tree of $F$ is an induced subgraph of $N$.
Note that, as in for directed networks, every forest-based unrooted phylogenetic 
network is tree-based, but that the converse may not hold. 
For example, we can take the network with vertex set 
$\{x,y,p,q,u,v\}$ and edge set $\{xp,pv,pu,uv,uq,vq,qy\}$,
which has leaf set $\{x,y\}$ and has two possible spanning trees
with leaf set $\{x,y\}$, namely the paths $x,p,v,u,q,y$ and $x,p,u,v,q,y$, neither
of which are induced paths.

\begin{figure}[htb]
    \centering
    \includegraphics[width=0.8\textwidth]{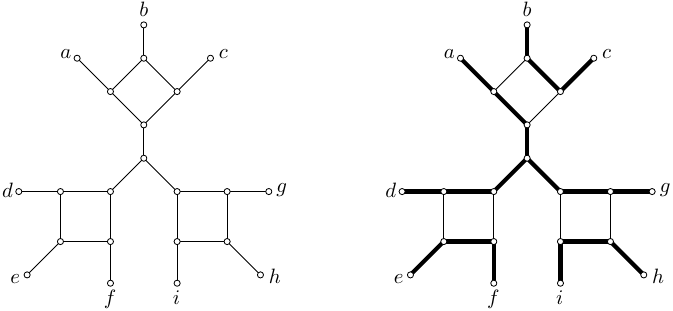}
    \caption{An example of a forest-based network that does not admit a leaf IPP.}
    \label{fig:counter-ex-unrooted}
\end{figure}

Interestingly, if an unrooted phylogenetic network $N$ is forest-based
then, in contrast to directed phylogenetic networks, it does  
not necessarily follow that $N$ contains an induced spanning forest 
with leaf set $L(N)$ that is the union of induced paths.
Notice that in an undirected induced path, the two endpoints of the path are its leaves, unlike the directed induced paths which only contain one leaf (the vertex of outdegree $0$).  The analogous notion of leaf IPP in undirected graphs therefore requires that each path has its \emph{two} endpoints in $L(N)$ (unless the path consists of a single vertex).
Consider for example the unrooted network shown in Figure~\ref{fig:counter-ex-unrooted} on the left.  It contains an induced spanning forest, as shown on the right.  However, it is not too difficult to verify that this network contains no leaf IPP.  Indeed, if we assume that such a leaf IPP exists, the central vertex of this network is adjacent to two vertices.  Thus, one of the neighboring subnetworks of that central vertex must itself admit a leaf IPP, which can be seen to be impossible.

It could be interesting to characterize forest-based 
unrooted phylogenetic networks that do have this property.
Note that if $N$ is tree-based,
then it does have a path partition whose paths all end
in $L(N)$, since we can clearly partition any subdivision tree 
into paths having this property.

Despite the above observation concerning unrooted 
forest-based networks, we can still use path partitions to show that it is NP-complete to
decide whether or not an unrooted network is forest-based as follows.
Suppose that $G=(V,E)$ is a connected, undirected graph. We say that $G$ has 
an {\em induced path partition} 
if its vertex set can be partitioned into
a collection of vertex-disjoint, induced paths in $G$.
In addition, we say that an unrooted phylogenetic network $N$ has 
a {\em leaf induced path partition} if it has an induced path partition
such that every path of length zero in the partition is contained in $L(N)$, and every 
other path in the partition intersects $L(N)$ precisely in its two end vertices. 
Note that any phylogenetic tree has such a partition, and
that path partitions arise in phylogenetic trees 
where they have applications to 
the so-called {\em phylogenetic targeting problem}  \cite{arnold2010polynomial}.

Although unrooted forest-based networks do not necessarily correspond to those admitting a leaf IPP, we show that this holds when the network has four leaves.

\begin{lemma}\label{lem:unrootpipp2}
   Suppose that $N$ is an  unrooted phylogenetic network
    with four leaves that is not a tree. Then $N$ is forest-based if and only 
    if $N$ has a leaf induced path partition containing two paths.
\end{lemma}
\begin{proof} 
If $N$ has a leaf induced path partition containing two paths, then
clearly $N$ is forest-based.

Conversely, suppose that $N$ is forest-based,
and that $F$ is an induced spanning forest in $N$ with 
leaf set $L(N)$.  Consider the number of connected components of $F$.
We see that $F$ cannot contain four connected components, since these could only 
be four paths of length $0$, all being elements of $L(N)$ (and since unrooted 
networks are connected by definition, there must be at least one vertex other than the leaves).

So, suppose that $F$ contains three connected components. Then two of
these components must be paths of length 0 (i.e. elements in $L(N)$)
and one of the components is an induced path $P$. Now, as $N$ is not a tree
it contains a cycle. But then every vertex in the cycle 
must be contained in the path $P$, as $P$ contains all vertices except two leaves, which is impossible as it
would contradict $P$ being an induced path.

Now, suppose that $F$ contains two connected components.
If these two components are paths, then 
$N$ has a perfect induced path partition containing two paths.
Otherwise, one of the components is an element in $L(N)$.
But then the other component in $F$ must be a tree with
three leaves, and it can be seen that this is not possible 
using a similar argument to the one used in the last paragraph (that is, all of 
the cycles in $N$ must be in that tree, a contradiction).

Finally, again using a similar argument, it follows that 
since $N$ is not a tree, $F$ cannot contain one connected component.
\end{proof}

In the following
we will make use of the proof of the following result \cite[Theorem 1]{le2003splitting} which 
we state using our terminology.

\begin{theorem}\label{thm:splitting}
Suppose that $G$ is an undirected graph.
Then it is NP-complete to decide whether or not $G$
has an induced path partition containing precisely two paths.
\end{theorem}

More specifically, we will make use of the difficult instances defined in the proof of
\cite[Theorem 1]{le2003splitting} which,  as can quickly be seen by inspecting the 
construction, consist of graphs with minimum degree at least $3$.  
Although the difficult graphs from Theorem~\ref{thm:splitting} do not have leaves, we can argue that if the forest-based recognition problem admitted a polynomial-time algorithm, we could call it multiple times to determine whether such a graph $G$ could be split into two induced paths, by adding four extra leaves at every possible location.  Recall that such a reduction, that requires multiple calls to a supposed polynomial time algorithm, is called a \emph{Turing reduction}.

\begin{theorem}\label{thm:unroothard}
    It is NP-complete (under Turing reductions) to decide whether or not an  
    unrooted phylogenetic network is forest-based.
\end{theorem}

\begin{proof}
First note that the problem is in NP, since a forest can serve as a certificate that can be verified in polynomial time.
We next show that the problem is NP-complete under Turing reductions, via the problem of partitioning an undirected graph into two induced paths, which we call the 2-path partition problem for the duration of the proof.  Recall that to achieve this, we assume access to a polynomial-time algorithm $A$ that can recognize unrooted forest-based networks, and show that this can be used to solve the 2-path partition problem in polynomial time.

Let $G$ be an instance of 2-path partition, where $G$ is assumed to be of minimum degree at least $3$.  In particular, $G$ has no leaves.
We may assume that for every vertex $v \in V(G)$, the graph $G - v$ obtained by removing $v$ is not an induced path, since such instances are easy to recognize in polynomial time.  Therefore, if $G$ can be partitioned into two induced paths, these paths have at least two vertices.
Let $Q = \{w,x,y,z\}$ be a set of four distinct vertices of $G$.  
Define the graph $G(Q)$ as follows: for every $u \in Q$, create a new vertex $u'$, and add the edge $u' u$. 
In other words, attach new leaves adjacent to $w', x', y', z'$ to $w, x, y, z$, respectively. 

For each subset $Q$ of four distinct vertices of $G$, execute $A$ on input $G(Q)$.  
If there is at least one $G(Q)$ that is forest-based according to $A$, then we return that $G$ can be partitioned into two induced paths.  Otherwise we return that no such partition exists.  

Clearly, this procedure runs in polynomial time if $A$ does run in polynomial time.  
We argue that it decides the instance $G$ correctly, by showing that $G$ admits an induced 2-path partition if and only if at least one $G(Q)$ is forest-based.
Suppose that $G$ can be split into two induced paths $P_1, P_2$.  
By our previous remark, $P_1$ and $P_2$ have at least two vertices each.
Let $w, x$ (resp. $y, z$) be the ends, i.e. the vertices of degree $1$, in $P_1$ (resp. in $P_2$).
Let $Q = \{w, x, y, z\}$.  Then $G(Q)$ admits a perfect induced path partition, namely $P_1 \cup \{w', x'\}$ and $P_2 \cup \{y', z'\}$, because extending the ends of the paths with an extra degree $1$ vertex preserves the induced property.  
Thus the above procedure correctly returns yes.  

Conversely, suppose that $G(Q)$ is forest-based for some $Q = \{w,x,y,z\}$.  
Note that because $G$ is assumed 
to have minimum degree 3, $w', x', y', z'$ are the only leaves of $G(Q)$
and $G(Q)$ is not a tree.  By Lemma~\ref{lem:unrootpipp2}, 
the vertices of $G(Q)$ can be split into two induced paths $P_1, P_2$, whose four ends are the leaves.  Say that the ends of $P_1$ are $w', x'$ and the ends of $P_2$ are $y', z'$.  Then $P_1 \setminus \{w', x'\}$ and $P_2 \setminus \{y', z'\}$ are induced paths of $G$.
\end{proof}

Observe that the hard instances generated in~\cite{le2003splitting} have unbounded degree.  The reduction is from NAE-3-SAT, and the maximum degree depends on the 
maximum number of occurrences of a variable in the Boolean formula.  It is plausible that by taking hard satisfiability instances with bounded variable occurrences, one could obtain hardness for induced 2-path partition with maximum degree bounded by a constant.  
However, this constant is likely to be higher than $3$, and novel ideas are needed to establish the complexity of recognizing \emph{binary} undirected forest-based networks.

\section{Discussion}

In this work, we have studied algorithmic problems of interest in two active research areas.  Indeed, forest-based networks and their variants will require further investigation in phylogenetics, whereas leaf induced path partitions give rise to novel problems in graph algorithms.
We were able to answer two open questions from both communities, namely that forest-based networks are hard to recognize, and that partitioning a binary DAG into a minimum number of induced paths is para-NP-hard.  Nonetheless, we have identified tractable instances that may be of use in practice, especially on orchard networks, and our results on unrooted phylogenetic networks pave the way for further exploration.  

Finally, throughout this paper we have encountered several problems that remain open,
as well as results which lead to some potential 
research directions. We conclude by summarising some of these: 
\begin{itemize}
    \item
    Recall that the \emph{level} of a network $N$ is the maximum number of reticulations in a biconnected component of $N$.  We observe that our difficult instances can have arbitrarily high levels.
    Is the forest-based recognition problem fixed-parameter tractable, when parameterized by the level of a network?  

    \item 
    Is the problem of finding a leaf IPP also NP-complete on networks with two leaves, but that are not required to be binary, in particular on networks of maximum total degree $4$?  

    \item 
    Is the forest-based recognition problem in P on superclasses of tree-child networks other than orchards, for instance 
    tree-sibling networks, reticulation-visible networks, or nearly stable networks?

    \item 
    Is it NP-complete to decide whether a \emph{binary} unrooted phylogenetic network is forest-based?

    \item In \cite{francis2018new} polynomial-time computable proximity-indices are introduced for measuring the
    extent to which an arbitrary binary phylogenetic network deviates from being tree-based. Unfortunately, in view of Theorem~\ref{thm:three-paths-hard}, 
    this approach does not directly extend to forest-based networks. Even so, it could still be interesting 
    to further study proximity measures for forest-based networks.

    \item There are interesting links between path partitions of digraphs and stable sets -- see e.g. \cite{sambinelli2022alpha}.
    It could
    interesting to study these concepts further for forest based networks.
    
\end{itemize}

\vspace{0.2cm}

\noindent{\bf Acknowledgement.}
The authors thank the Institute for Mathematical Sciences, National University of Singapore, 
for their invitation to attend the “Mathematics of Evolution - Phylogenetic Trees and Networks” program
in 2023, in which they began discussing the problems investigated in this paper.

\vspace{0.2cm}

\noindent{\bf Conflict of interest.}   The authors declare that they have no competing interests.

\bibliography{main}


\begin{thebibliography}{35}
\ifx \bisbn   \undefined \def \bisbn  #1{ISBN #1}\fi
\ifx \binits  \undefined \def \binits#1{#1}\fi
\ifx \bauthor  \undefined \def \bauthor#1{#1}\fi
\ifx \batitle  \undefined \def \batitle#1{#1}\fi
\ifx \bjtitle  \undefined \def \bjtitle#1{#1}\fi
\ifx \bvolume  \undefined \def \bvolume#1{\textbf{#1}}\fi
\ifx \byear  \undefined \def \byear#1{#1}\fi
\ifx \bissue  \undefined \def \bissue#1{#1}\fi
\ifx \bfpage  \undefined \def \bfpage#1{#1}\fi
\ifx \blpage  \undefined \def \blpage #1{#1}\fi
\ifx \burl  \undefined \def \burl#1{\textsf{#1}}\fi
\ifx \doiurl  \undefined \def \doiurl#1{\url{https://doi.org/#1}}\fi
\ifx \betal  \undefined \def \betal{\textit{et al.}}\fi
\ifx \binstitute  \undefined \def \binstitute#1{#1}\fi
\ifx \binstitutionaled  \undefined \def \binstitutionaled#1{#1}\fi
\ifx \bctitle  \undefined \def \bctitle#1{#1}\fi
\ifx \beditor  \undefined \def \beditor#1{#1}\fi
\ifx \bpublisher  \undefined \def \bpublisher#1{#1}\fi
\ifx \bbtitle  \undefined \def \bbtitle#1{#1}\fi
\ifx \bedition  \undefined \def \bedition#1{#1}\fi
\ifx \bseriesno  \undefined \def \bseriesno#1{#1}\fi
\ifx \blocation  \undefined \def \blocation#1{#1}\fi
\ifx \bsertitle  \undefined \def \bsertitle#1{#1}\fi
\ifx \bsnm \undefined \def \bsnm#1{#1}\fi
\ifx \bsuffix \undefined \def \bsuffix#1{#1}\fi
\ifx \bparticle \undefined \def \bparticle#1{#1}\fi
\ifx \barticle \undefined \def \barticle#1{#1}\fi
\bibcommenthead
\ifx \bconfdate \undefined \def \bconfdate #1{#1}\fi
\ifx \botherref \undefined \def \botherref #1{#1}\fi
\ifx \url \undefined \def \url#1{\textsf{#1}}\fi
\ifx \bchapter \undefined \def \bchapter#1{#1}\fi
\ifx \bbook \undefined \def \bbook#1{#1}\fi
\ifx \bcomment \undefined \def \bcomment#1{#1}\fi
\ifx \oauthor \undefined \def \oauthor#1{#1}\fi
\ifx \citeauthoryear \undefined \def \citeauthoryear#1{#1}\fi
\ifx \endbibitem  \undefined \def \endbibitem {}\fi
\ifx \bconflocation  \undefined \def \bconflocation#1{#1}\fi
\ifx \arxivurl  \undefined \def \arxivurl#1{\textsf{#1}}\fi
\csname PreBibitemsHook\endcsname

\bibitem[\protect\citeauthoryear{Kong et~al.}{2022}]{kong2022classes}
\begin{barticle}
\bauthor{\bsnm{Kong}, \binits{S.}},
\bauthor{\bsnm{Pons}, \binits{J.C.}},
\bauthor{\bsnm{Kubatko}, \binits{L.}},
\bauthor{\bsnm{Wicke}, \binits{K.}}:
\batitle{Classes of explicit phylogenetic networks and their biological and mathematical significance}.
\bjtitle{Journal of Mathematical Biology}
\bvolume{84}(\bissue{6}),
\bfpage{47}
(\byear{2022})
\end{barticle}
\endbibitem

\bibitem[\protect\citeauthoryear{Scholz et~al.}{2019}]{scholz2019osf}
\begin{barticle}
\bauthor{\bsnm{Scholz}, \binits{G.E.}},
\bauthor{\bsnm{Popescu}, \binits{A.-A.}},
\bauthor{\bsnm{Taylor}, \binits{M.I.}},
\bauthor{\bsnm{Moulton}, \binits{V.}},
\bauthor{\bsnm{Huber}, \binits{K.T.}}:
\batitle{{OSF}-builder: a new tool for constructing and representing evolutionary histories involving introgression}.
\bjtitle{Systematic Biology}
\bvolume{68}(\bissue{5}),
\bfpage{717}--\blpage{729}
(\byear{2019})
\end{barticle}
\endbibitem

\bibitem[\protect\citeauthoryear{Sneath}{1975}]{sneath1975cladistic}
\begin{barticle}
\bauthor{\bsnm{Sneath}, \binits{P.H.}}:
\batitle{Cladistic representation of reticulate evolution}.
\bjtitle{Systematic Zoology}
\bvolume{24}(\bissue{3}),
\bfpage{360}--\blpage{368}
(\byear{1975})
\end{barticle}
\endbibitem

\bibitem[\protect\citeauthoryear{Steel}{2016}]{steel2016phylogeny}
\begin{bbook}
\bauthor{\bsnm{Steel}, \binits{M.}}:
\bbtitle{Phylogeny: Discrete and Random Processes in Evolution}.
\bpublisher{SIAM},
\blocation{Philadelphia}
(\byear{2016})
\end{bbook}
\endbibitem

\bibitem[\protect\citeauthoryear{Francis and Steel}{2015}]{francis2015phylogenetic}
\begin{barticle}
\bauthor{\bsnm{Francis}, \binits{A.R.}},
\bauthor{\bsnm{Steel}, \binits{M.}}:
\batitle{Which phylogenetic networks are merely trees with additional arcs?}
\bjtitle{Systematic Biology}
\bvolume{64}(\bissue{5}),
\bfpage{768}--\blpage{777}
(\byear{2015})
\end{barticle}
\endbibitem

\bibitem[\protect\citeauthoryear{Jetten and van Iersel}{2016}]{jetten2016nonbinary}
\begin{barticle}
\bauthor{\bsnm{Jetten}, \binits{L.}},
\bauthor{\bsnm{Iersel}, \binits{L.}}:
\batitle{Nonbinary tree-based phylogenetic networks}.
\bjtitle{IEEE/ACM Transactions on Computational Biology and Bioinformatics}
\bvolume{15}(\bissue{1}),
\bfpage{205}--\blpage{217}
(\byear{2016})
\end{barticle}
\endbibitem

\bibitem[\protect\citeauthoryear{Francis et~al.}{2018}]{francis2018new}
\begin{barticle}
\bauthor{\bsnm{Francis}, \binits{A.}},
\bauthor{\bsnm{Semple}, \binits{C.}},
\bauthor{\bsnm{Steel}, \binits{M.}}:
\batitle{New characterisations of tree-based networks and proximity measures}.
\bjtitle{Advances in Applied Mathematics}
\bvolume{93},
\bfpage{93}--\blpage{107}
(\byear{2018})
\end{barticle}
\endbibitem

\bibitem[\protect\citeauthoryear{Huber et~al.}{2022}]{huber2022forest}
\begin{barticle}
\bauthor{\bsnm{Huber}, \binits{K.T.}},
\bauthor{\bsnm{Moulton}, \binits{V.}},
\bauthor{\bsnm{Scholz}, \binits{G.E.}}:
\batitle{Forest-based networks}.
\bjtitle{Bulletin of Mathematical Biology}
\bvolume{84}(\bissue{10}),
\bfpage{119}
(\byear{2022})
\end{barticle}
\endbibitem

\bibitem[\protect\citeauthoryear{Manuel}{2018}]{manuel2018revisiting}
\begin{botherref}
\oauthor{\bsnm{Manuel}, \binits{P.}}:
Revisiting path-type covering and partitioning problems.
arXiv preprint arXiv:1807.10613
(2018)
\end{botherref}
\endbibitem

\bibitem[\protect\citeauthoryear{Fernau et~al.}{2023}]{fernau2023parameterizing}
\begin{bchapter}
\bauthor{\bsnm{Fernau}, \binits{H.}},
\bauthor{\bsnm{Foucaud}, \binits{F.}},
\bauthor{\bsnm{Mann}, \binits{K.}},
\bauthor{\bsnm{Padariya}, \binits{U.}},
\bauthor{\bsnm{Rao}, \binits{K.R.}}:
\bctitle{Parameterizing path partitions}.
In: \bbtitle{International Conference on Algorithms and Complexity},
pp. \bfpage{187}--\blpage{201}
(\byear{2023}).
\bcomment{Springer}
\end{bchapter}
\endbibitem

\bibitem[\protect\citeauthoryear{Sambinelli et~al.}{2022}]{sambinelli2022alpha}
\begin{barticle}
\bauthor{\bsnm{Sambinelli}, \binits{M.}},
\bauthor{\bsnm{Silva}, \binits{C.N.}},
\bauthor{\bsnm{Lee}, \binits{O.}}:
\batitle{$\alpha$-diperfect digraphs}.
\bjtitle{Discrete Mathematics}
\bvolume{345}(\bissue{5}),
\bfpage{112759}
(\byear{2022})
\end{barticle}
\endbibitem

\bibitem[\protect\citeauthoryear{Huber et~al.}{In press}]{huber2023network}
\begin{botherref}
\oauthor{\bsnm{Huber}, \binits{K.T.}},
\oauthor{\bsnm{Iersel}, \binits{L.}},
\oauthor{\bsnm{Moulton}, \binits{V.}},
\oauthor{\bsnm{Scholz}, \binits{G.}}:
Is this network proper forest-based?
Information Processing Letters
(In press)
\end{botherref}
\endbibitem

\bibitem[\protect\citeauthoryear{Dehghan et~al.}{2015}]{dehghan2015complexity}
\begin{barticle}
\bauthor{\bsnm{Dehghan}, \binits{A.}},
\bauthor{\bsnm{Sadeghi}, \binits{M.-R.}},
\bauthor{\bsnm{Ahadi}, \binits{A.}}:
\batitle{On the complexity of deciding whether the regular number is at most two}.
\bjtitle{Graphs and Combinatorics}
\bvolume{31}(\bissue{5}),
\bfpage{1359}--\blpage{1365}
(\byear{2015})
\end{barticle}
\endbibitem

\bibitem[\protect\citeauthoryear{Impagliazzo et~al.}{2001}]{impagliazzo2001problems}
\begin{barticle}
\bauthor{\bsnm{Impagliazzo}, \binits{R.}},
\bauthor{\bsnm{Paturi}, \binits{R.}},
\bauthor{\bsnm{Zane}, \binits{F.}}:
\batitle{Which problems have strongly exponential complexity?}
\bjtitle{Journal of Computer and System Sciences}
\bvolume{63}(\bissue{4}),
\bfpage{512}--\blpage{530}
(\byear{2001})
\end{barticle}
\endbibitem

\bibitem[\protect\citeauthoryear{Antony et~al.}{2024}]{antony2024switching}
\begin{botherref}
\oauthor{\bsnm{Antony}, \binits{D.}},
\oauthor{\bsnm{Cao}, \binits{Y.}},
\oauthor{\bsnm{Pal}, \binits{S.}},
\oauthor{\bsnm{Sandeep}, \binits{R.}}:
Switching classes: Characterization and computation.
arXiv preprint arXiv:2403.04263
(2024)
\end{botherref}
\endbibitem

\bibitem[\protect\citeauthoryear{Erd{\H{o}}s et~al.}{2019}]{erdHos2019class}
\begin{barticle}
\bauthor{\bsnm{Erd{\H{o}}s}, \binits{P.L.}},
\bauthor{\bsnm{Semple}, \binits{C.}},
\bauthor{\bsnm{Steel}, \binits{M.}}:
\batitle{A class of phylogenetic networks reconstructable from ancestral profiles}.
\bjtitle{Mathematical Biosciences}
\bvolume{313},
\bfpage{33}--\blpage{40}
(\byear{2019})
\end{barticle}
\endbibitem

\bibitem[\protect\citeauthoryear{van Iersel et~al.}{2022}]{van2022orchard}
\begin{barticle}
\bauthor{\bsnm{Iersel}, \binits{L.}},
\bauthor{\bsnm{Janssen}, \binits{R.}},
\bauthor{\bsnm{Jones}, \binits{M.}},
\bauthor{\bsnm{Murakami}, \binits{Y.}}:
\batitle{Orchard networks are trees with additional horizontal arcs}.
\bjtitle{Bulletin of Mathematical Biology}
\bvolume{84}(\bissue{8}),
\bfpage{76}
(\byear{2022})
\end{barticle}
\endbibitem

\bibitem[\protect\citeauthoryear{Cardona et~al.}{2024}]{cardona2024comparison}
\begin{botherref}
\oauthor{\bsnm{Cardona}, \binits{G.}},
\oauthor{\bsnm{Pons}, \binits{J.C.}},
\oauthor{\bsnm{Ribas}, \binits{G.}},
\oauthor{\bsnm{Coronado}, \binits{T.M.}}:
Comparison of orchard networks using their extended $\mu$-representation.
IEEE/ACM Transactions on Computational Biology and Bioinformatics
(2024)
\end{botherref}
\endbibitem

\bibitem[\protect\citeauthoryear{Landry et~al.}{2023}]{landry2023fixed}
\begin{botherref}
\oauthor{\bsnm{Landry}, \binits{K.}},
\oauthor{\bsnm{Tremblay-Savard}, \binits{O.}},
\oauthor{\bsnm{Lafond}, \binits{M.}}:
A fixed-parameter tractable algorithm for finding agreement cherry-reduced subnetworks in level-1 orchard networks.
Journal of Computational Biology
(2023)
\end{botherref}
\endbibitem

\bibitem[\protect\citeauthoryear{Janssen and Murakami}{2021}]{janssen2021cherry}
\begin{barticle}
\bauthor{\bsnm{Janssen}, \binits{R.}},
\bauthor{\bsnm{Murakami}, \binits{Y.}}:
\batitle{On cherry-picking and network containment}.
\bjtitle{Theoretical Computer Science}
\bvolume{856},
\bfpage{121}--\blpage{150}
(\byear{2021})
\end{barticle}
\endbibitem

\bibitem[\protect\citeauthoryear{van Bevern et~al.}{2017}]{van2017fixed}
\begin{barticle}
\bauthor{\bsnm{Bevern}, \binits{R.}},
\bauthor{\bsnm{Bredereck}, \binits{R.}},
\bauthor{\bsnm{Chopin}, \binits{M.}},
\bauthor{\bsnm{Hartung}, \binits{S.}},
\bauthor{\bsnm{H{\"u}ffner}, \binits{F.}},
\bauthor{\bsnm{Nichterlein}, \binits{A.}},
\bauthor{\bsnm{Such{\`y}}, \binits{O.}}:
\batitle{Fixed-parameter algorithms for {DAG} partitioning}.
\bjtitle{Discrete Applied Mathematics}
\bvolume{220},
\bfpage{134}--\blpage{160}
(\byear{2017})
\end{barticle}
\endbibitem

\bibitem[\protect\citeauthoryear{Fortune et~al.}{1980}]{fortune1980directed}
\begin{barticle}
\bauthor{\bsnm{Fortune}, \binits{S.}},
\bauthor{\bsnm{Hopcroft}, \binits{J.}},
\bauthor{\bsnm{Wyllie}, \binits{J.}}:
\batitle{The directed subgraph homeomorphism problem}.
\bjtitle{Theoretical Computer Science}
\bvolume{10}(\bissue{2}),
\bfpage{111}--\blpage{121}
(\byear{1980})
\end{barticle}
\endbibitem

\bibitem[\protect\citeauthoryear{Tholey}{2012}]{tholey2012linear}
\begin{barticle}
\bauthor{\bsnm{Tholey}, \binits{T.}}:
\batitle{Linear time algorithms for two disjoint paths problems on directed acyclic graphs}.
\bjtitle{Theoretical Computer Science}
\bvolume{465},
\bfpage{35}--\blpage{48}
(\byear{2012})
\end{barticle}
\endbibitem

\bibitem[\protect\citeauthoryear{Kawarabayashi and Kobayashi}{2008}]{kawarabayashi2008induced}
\begin{bchapter}
\bauthor{\bsnm{Kawarabayashi}, \binits{K.-i.}},
\bauthor{\bsnm{Kobayashi}, \binits{Y.}}:
\bctitle{The induced disjoint paths problem}.
In: \bbtitle{Integer Programming and Combinatorial Optimization: 13th International Conference, IPCO 2008 Bertinoro, Italy, May 26-28, 2008 Proceedings 13},
pp. \bfpage{47}--\blpage{61}
(\byear{2008}).
\bcomment{Springer}
\end{bchapter}
\endbibitem

\bibitem[\protect\citeauthoryear{Slivkins}{2010}]{slivkins2010parameterized}
\begin{barticle}
\bauthor{\bsnm{Slivkins}, \binits{A.}}:
\batitle{Parameterized tractability of edge-disjoint paths on directed acyclic graphs}.
\bjtitle{SIAM Journal on Discrete Mathematics}
\bvolume{24}(\bissue{1}),
\bfpage{146}--\blpage{157}
(\byear{2010})
\end{barticle}
\endbibitem

\bibitem[\protect\citeauthoryear{B{\'e}rczi and Kobayashi}{2017}]{berczi2017directed}
\begin{bchapter}
\bauthor{\bsnm{B{\'e}rczi}, \binits{K.}},
\bauthor{\bsnm{Kobayashi}, \binits{Y.}}:
\bctitle{The directed disjoint shortest paths problem}.
In: \bbtitle{25th Annual European Symposium on Algorithms (ESA 2017)}
(\byear{2017}).
\bcomment{Schloss-Dagstuhl-Leibniz Zentrum f{\"u}r Informatik}
\end{bchapter}
\endbibitem

\bibitem[\protect\citeauthoryear{Lopes and Sau}{2022}]{lopes2022relaxation}
\begin{barticle}
\bauthor{\bsnm{Lopes}, \binits{R.}},
\bauthor{\bsnm{Sau}, \binits{I.}}:
\batitle{A relaxation of the directed disjoint paths problem: A global congestion metric helps}.
\bjtitle{Theoretical Computer Science}
\bvolume{898},
\bfpage{75}--\blpage{91}
(\byear{2022})
\end{barticle}
\endbibitem

\bibitem[\protect\citeauthoryear{Huber et~al.}{2022}]{huber2022overlaid}
\begin{barticle}
\bauthor{\bsnm{Huber}, \binits{K.T.}},
\bauthor{\bsnm{Moulton}, \binits{V.}},
\bauthor{\bsnm{Scholz}, \binits{G.E.}}:
\batitle{Overlaid species forests}.
\bjtitle{Discrete Applied Mathematics}
\bvolume{309},
\bfpage{110}--\blpage{122}
(\byear{2022})
\end{barticle}
\endbibitem

\bibitem[\protect\citeauthoryear{Cormen et~al.}{2022}]{cormen2022introduction}
\begin{bbook}
\bauthor{\bsnm{Cormen}, \binits{T.H.}},
\bauthor{\bsnm{Leiserson}, \binits{C.E.}},
\bauthor{\bsnm{Rivest}, \binits{R.L.}},
\bauthor{\bsnm{Stein}, \binits{C.}}:
\bbtitle{Introduction to Algorithms}.
\bpublisher{MIT press},
\blocation{Cambridge, Massachusetts}
(\byear{2022})
\end{bbook}
\endbibitem

\bibitem[\protect\citeauthoryear{Hopcroft and Karp}{1973}]{hopcroft1973n}
\begin{barticle}
\bauthor{\bsnm{Hopcroft}, \binits{J.E.}},
\bauthor{\bsnm{Karp}, \binits{R.M.}}:
\batitle{An $n^{\frac{5}{2}}$ algorithm for maximum matchings in bipartite graphs}.
\bjtitle{SIAM Journal on Computing}
\bvolume{2}(\bissue{4}),
\bfpage{225}--\blpage{231}
(\byear{1973})
\end{barticle}
\endbibitem

\bibitem[\protect\citeauthoryear{Aspvall et~al.}{1979}]{aspvall1979linear}
\begin{barticle}
\bauthor{\bsnm{Aspvall}, \binits{B.}},
\bauthor{\bsnm{Plass}, \binits{M.F.}},
\bauthor{\bsnm{Tarjan}, \binits{R.E.}}:
\batitle{A linear-time algorithm for testing the truth of certain quantified boolean formulas}.
\bjtitle{Information Processing Letters}
\bvolume{8}(\bissue{3}),
\bfpage{121}--\blpage{123}
(\byear{1979})
\end{barticle}
\endbibitem

\bibitem[\protect\citeauthoryear{Francis et~al.}{2018}]{francis2018tree}
\begin{barticle}
\bauthor{\bsnm{Francis}, \binits{A.}},
\bauthor{\bsnm{Huber}, \binits{K.T.}},
\bauthor{\bsnm{Moulton}, \binits{V.}}:
\batitle{Tree-based unrooted phylogenetic networks}.
\bjtitle{Bulletin of Mathematical Biology}
\bvolume{80},
\bfpage{404}--\blpage{416}
(\byear{2018})
\end{barticle}
\endbibitem

\bibitem[\protect\citeauthoryear{Hendriksen}{2018}]{hendriksen2018tree}
\begin{barticle}
\bauthor{\bsnm{Hendriksen}, \binits{M.}}:
\batitle{Tree-based unrooted nonbinary phylogenetic networks}.
\bjtitle{Mathematical Biosciences}
\bvolume{302},
\bfpage{131}--\blpage{138}
(\byear{2018})
\end{barticle}
\endbibitem

\bibitem[\protect\citeauthoryear{Arnold and Stadler}{2010}]{arnold2010polynomial}
\begin{barticle}
\bauthor{\bsnm{Arnold}, \binits{C.}},
\bauthor{\bsnm{Stadler}, \binits{P.F.}}:
\batitle{Polynomial algorithms for the maximal pairing problem: efficient phylogenetic targeting on arbitrary trees}.
\bjtitle{Algorithms for Molecular Biology}
\bvolume{5}(\bissue{1}),
\bfpage{1}--\blpage{10}
(\byear{2010})
\end{barticle}
\endbibitem

\bibitem[\protect\citeauthoryear{Le et~al.}{2003}]{le2003splitting}
\begin{barticle}
\bauthor{\bsnm{Le}, \binits{H.-O.}},
\bauthor{\bsnm{M{\"u}ller}, \binits{H.}}, \betal:
\batitle{Splitting a graph into disjoint induced paths or cycles}.
\bjtitle{Discrete Applied Mathematics}
\bvolume{131}(\bissue{1}),
\bfpage{199}--\blpage{212}
(\byear{2003})
\end{barticle}
\endbibitem

\end{thebibliography}

\end{document}